\documentclass[11pt]{article}
\usepackage{amsmath, amsthm, amssymb}
\usepackage{fullpage}
\usepackage{tikz}
\usetikzlibrary{decorations.pathreplacing}

\newcommand{\ket}[1]{|#1\rangle}

\newcommand{\circuit}[1]{\mathcal{#1}}

\newtheorem{thm}{Theorem}

\newtheorem{lem}[thm]{Lemma}
\theoremstyle{definition}
\newtheorem{dfn}{Definition}

\newcounter{conditionnum}
\renewcommand{\theconditionnum}{\roman{conditionnum}}
\newenvironment{conditions}{\begin{list}{(\theconditionnum)}%
{\setlength{\itemindent}{-1.5ex}\usecounter{conditionnum}}}{\end{list}}

\DeclareMathOperator{\polylog}{polylog}
\DeclareMathOperator{\poly}{poly}
\DeclareMathOperator{\wt}{wt}

\begin{document}

\author{Daniel Gottesman\thanks{E-mail: dgottesman@perimeterinstitute.ca} \\  \\
\begin{tabular}{c} Perimeter Institute \\ Waterloo, Canada \end{tabular}
\begin{tabular}{c} CIFAR QIS Program \\ Toronto, Canada \end{tabular}
% M5G 1Z8
}

\title{Fault-Tolerant Quantum Computation with Constant Overhead}

\date{}

\maketitle

\begin{abstract}
What is the minimum number of extra qubits needed to perform a large fault-tolerant quantum circuit?  Working in a common model of fault-tolerance, I show that  in the asymptotic limit of large circuits, the ratio of physical qubits to logical qubits can be a constant.  The construction makes use of quantum low-density parity check codes, and the asymptotic overhead of the protocol is equal to that of the family of quantum error-correcting codes underlying the fault-tolerant protocol.
\end{abstract}

\section{Introduction}
\label{sec:intro}

In order to build a large quantum computer, it will probably be necessary to use some sort of fault-tolerant protocol.  It seems implausible that we can eliminate errors in an experimental system sufficiently well to perform a long quantum computation without a single error occurring during its course.  Therefore, some additional method of controlling errors is needed, and many have been devised.  The most general technique known is to encode the data qubits used in the computation in a quantum error-correcting code and to use fault-tolerant gate constructions to perform gates on the data without ever leaving the protective encoding.  The threshold theorem~\cite{AB97,Kit97,KLZ98} shows that using these techniques allows us to perform arbitrarily long quantum computations provided the physical error rate per gate or time step of the computation is below some constant threshold value.  The threshold theorem applies to very general types of weak noise~\cite{AGP06,TB05}, and therefore serves as a solution of last resort to the problem of errors when more specialized techniques have reached their limits.

In the standard version of the threshold theorem, a logical circuit using $m$ qubits and containing $T$ gates is replaced by a fault-tolerant circuit using $O(m \polylog (mT))$ qubits.  Asymptotically, this is a very favorable scaling of overhead (ratio of physical qubits to logical qubits), but in practice the constant factors hidden by the big-O notation are quite large, ranging from hundreds or thousands for the most efficient known protocols to billions for protocols maximally optimized for high threshold~\cite{Knill}.  The error rates in the best qubits that have been created are starting to approach the level needed for fault-tolerant protocols, but the number of qubits that can be realized reliably is still quite small.  The high overhead cost of fault-tolerant protocols exerts a large tax on future quantum computers, and with current protocols, the vast majority of the qubits in the computer would simply be used for correcting errors.

Ultimately, the cause of the high overhead of fault-tolerant protocols is that in many protocols, error correction for each qubit is treated separately.  For instance, with concatenated codes, each logical qubit is encoded in a separate block of the code.  The code must be able to correct multiple errors, since in a long computation, even with a low error rate, there will always be fluctuations where many errors occur at the same time, and there is a cost associated with the ability to correct extra errors.  Furthermore, associated to each block of the code are a large number of ancilla qubits used to perform fault-tolerant error correction and some types of fault-tolerant gates.  Modern techniques for surface codes do encode multiple logical qubits in the same block of the code, but keep them far apart from each other, leading to similar considerations: Each logical qubit has associated with it many physical qubits devoted to protecting that logical qubit from errors.  

Taking these protocols as models, one might be tempted to conjecture that there must necessarily be some tradeoff between doing long logical computations and overhead: since the code must be able to correct more errors for a long computation, perhaps it is necessary to have more overhead for a long computation.  However, note that there is no such tradeoff for pure quantum error correction: If gate errors are not a problem, so we only need to correct errors that occur during transmission through a noisy quantum channel, then it is possible to send $k$ logical qubits using $n$ physical qubits with a constant error rate $p$ per qubit sent through the channel, and yet have $k/n$ approach a constant rate $R$ for large $n$.  $R$ is the \emph{channel capacity} of the noisy communications channel.  Furthermore, while there is a tradeoff between error rate $p$ and data rate $R$, a constant $R$ is achievable even for relatively high error rates (around $20\%$ for the depolarizing channel).

Efficient error correction can be achieved using large codes that encode many qubits in a single block and correct them together as a unit.  However, there has been little previous work on fault-tolerance with a family of codes of this type.  \cite{Got97} showed that fault-tolerant protocols exist for stabilizer codes encoding multiple qubits per block, with further improvements in \cite{SI03}.  Steane~\cite{Steane} investigated using such codes for computations of a reasonable size, and indeed showed that it is possible to substantially reduce the overhead requirements of a fault-tolerant protocol.  However, Steane did not show a threshold result: He analyzed fault tolerance for a few specific codes, and in a sufficiently long computation, those codes will fail.

The difficulty in getting a threshold result while still maintaining a modest overhead is that to achieve efficient error correction in the asymptotic limit, the codes we use must get bigger as the number of logical qubits increases.  This presents a challenge for fault-tolerant protocols, since they typically use ancilla states which involve a number of physical qubits comparable to the size of the code.  (Indeed, often the ancilla states used are particular states encoded in another block of the same code used to protect the data.)  Creating such a large ancilla state is a challenge in the presence of error, and doing so without producing too much extra overhead is even more difficult.

In this paper, I show that, given a family of quantum error-correcting codes with the right properties, it is nonetheless possible to create fault-tolerant protocols that allow arbitrarily large quantum computations with constant error rate and constant overhead.  The overhead needed is asymptotically equal to $1/R$, the overhead (inverse rate) of the code family.  Using low-density parity check (LDPC) codes allows fault-tolerant error correction with small ancillas.  By choosing code blocks of the right size and by controlling the flow of ancilla states appropriately, it is possible to keep the overhead under control, giving the desired result.  Indeed, by choosing a suitable family of codes, $R$ can be made as close to $1$ as desired.  This shows that the overhead needed can be made as small as desired provided the physical error rate is low enough.

Recent work by Kovalev and Pryadko~\cite{KP12b} has shown that codes with all the necessary properties exist.  However, while the main result of this paper applies to various known code families~\cite{FML,GL,Has13,TZ}, none of the families is quite satisfactory.  Some of these families lack an efficient classical syndrome decoding algorithm, so the result only applies if classical computation, even exponentially long classical computation, is assumed to be free.  Other families have an efficient decoding algorithm, but do not suppress error rates as much as we would want, meaning that the threshold depends on the difficulty of the computation being performed.  However, the main theorem is not tied to any specific family of codes, so if a new code family is discovered satisfying the conditions of Thm.~\ref{thm:main} (or a better decoding algorithm is found for a known code), then it can be immediately used with the fault-tolerant protocols of this paper.

The result of this paper is primarily an asymptotic result.  It applies only to very large computations, as there are sub-leading additive terms in the overhead which might be quite important for small computations.  However, it illustrates the advantages of fault tolerance using codes encoding multiple qubits per block.  I hope it will spur further investigation into fault-tolerant protocols for such codes, and it may be that some of the techniques discussed in this paper will be useful for reducing the overhead in small systems.

I will begin in Sec.~\ref{sec:basicFT} with a brief introduction to quantum error-correcting codes and a discussion of the model of fault tolerant circuits that I will be using, and then state the main theorem of the paper in Sec.~\ref{sec:main}.  In Sec.~\ref{sec:LDPC}, I discuss some properties of the families of LDPC codes that satisfy the assumptions of the main theorem.  I then show in Secs.~\ref{sec:FTEC} and \ref{sec:gates} how to perform fault-tolerant error correction and the other components of a fault-tolerant circuit without producing too much overhead, and the components are combined in Sec.~\ref{sec:combine} in order to prove the main theorem.  In Sec.~\ref{sec:depth}, I will comment on the depth blow-up created by the protocol, and I will conclude in Sec.~\ref{sec:conclusion}.

\section{Basic Model of Fault Tolerance}
\label{sec:basicFT}

In a fault-tolerant protocol for quantum computation, some number of \emph{logical qubits} are encoded in a quantum error-correcting code (QECC) using a larger number of \emph{physical qubits}.  In many fault-tolerant protocols, each logical qubit is encoded in a separate \emph{block} of the code, but in this paper, I will consider protocols for which each block of the code contains many logical qubits.  A QECC  with $k$ logical qubits encoded in $n$ physical qubits can be defined as a $2^k$-dimensional subspace (the \emph{code space}) of the full $2^n$-dimensional physical Hilbert space.  See~\cite{Got09} for a more complete introduction to QECCs and fault-tolerant quantum computation.

All the codes I will discuss in this paper are stabilizer codes, which can be defined via a \emph{stabilizer}, an Abelian group of Pauli operators (tensor products of the single-qubit Paulis $X$, $Y$, $Z$, and $I$, with an overall phase of $\pm 1$ or $\pm i$).  The operators $M$ in the stabilizer all have the property that $M \ket{\psi} = \ket{\psi}$ for all codewords $\ket{\psi}$ in the code space.  The stabilizer of a code with $k$ logical qubits and $n$ physical qubits consists of $2^{n-k}$ elements, and can be generated by a set of $n-k$ elements.

The stabilizer is useful both as an efficient description of the code and because it describes how to identify errors.  If $Q$ is a stabilizer and $M \in Q$, then the correct codewords all have eigenvalue $+1$ for $M$.  All pairs of Pauli operators either commute or anti-commute.  If $E$ is a Pauli error that has occurred on the codeword $\ket{\psi}$, then $E \ket{\psi}$ has eigenvalue $+1$ for $M$ if $E$ and $M$ commute, and eigenvalue $-1$ for $M$ if $E$ and $M$ anti-commute.  Thus, the list of eigenvalues of generators of $Q$ for a potentially erroneous state gives a great deal of information about whether an error has occurred and what the nature of the error is.  The list of eigenvalues is called the \emph{error syndrome} of the error; usually we write it as a bit string, with $0$ representing eigenvalue $+1$ and $1$ representing eigenvalue $-1$.

We can define the set
\begin{equation}
N(Q) = \{ E | EM=ME \ \forall M \in Q\}.
\end{equation}
Two errors $E$ and $F$ have the same error syndrome iff $E^\dagger F \in N(Q)$. However, if $E^\dagger F \in Q \subseteq N(Q)$, then $E$ and $F$ have the same action on codewords.  If $E^\dagger F \not\in N(Q) \setminus Q$ for all pairs of possible errors $E$ and $F$, then the code can correct that set of errors by either distinguishing the error syndromes or just treating the errors the same way.  Therefore, for a stabilizer code, we define the distance $d$ as the minimum weight of $N \in N(Q) \setminus Q$.  A distance $d$ QECC can correct $\lfloor (d-1)/2 \rfloor$ errors. If the stabilizer $Q$ contains any elements with weight less than $d$, the code is a \emph{degenerate code}; otherwise it is non-degenerate.  A stabilizer code with $n$ physical qubits, $k$ logical qubits, and distance $d$ is referred to as an $[[n,k,d]]$ code.  I will sometimes omit $d$ from this notation when the distance of the code is unspecified.  

This paper will make use of a subset of stabilizer codes, known as \emph{low-density parity check} or \emph{LDPC} codes.  A quantum LDPC code is a stabilizer code with the additional property that all generators of the stabilizer are low weight.
\begin{dfn}
An $[[n,k]]$ stabilizer code is an $(r,c)$-LDPC code if there exists a choice of generators $\{M_1, \ldots, M_{n-k}\}$ for the stabilizer $Q$ of the code such that $\wt M_i \leq r$ for all $i$ and the number of generators which act non-trivially (i.e., as $X$, $Y$, or $Z$) on qubit $j$ is at most $c$ for all $j \in [1,n]$.
\end{dfn}
Technically, based on this definition, every $[[n,k]]$ stabilizer code is automatically an $(n,n-k)$-LDPC code.  However, typically, when we refer to an ``LDPC'' code, we are only considering the interesting cases, which arise when $r$ and $c$ are much smaller than $n$ and $n-k$, particularly when $r$ and $c$ are constant for large $n$.  Note that whenever the distance $d$ of the LDPC code is greater than $r$, the code is degenerate.

The definition only requires that some particular choice of generators for the stabilizer have these two properties.  Other choices of sets of generators for the same code will usually not satisfy these conditions.  Indeed, given a set of generators for $Q$, it might not be straightforward to find the correct set of generators that determines $r$ and $c$.  There might not also not be a unique choice of such generators, and there might be a trade-off between $r$ and $c$ for different sets of generators.  For these reasons, an LDPC code should be presented with a particular choice of generators that make the LDPC nature of the code explicit.

There is a similar definition of a classical $(r,c)$-LDPC code, using the parity check matrix in place of the stabilizer.

I will work with one of the most common models of fault-tolerant quantum computation, which one might call the \emph{basic model} of fault tolerance.  The basic model makes some simplifications which are inessential to the results but are very helpful to make the analysis tractable.  In particular, we typically assume that the circuit consists of a number of \emph{locations}.  A location could be a single-qubit gate, a multiple-qubit gate, the preparation of a single physical qubit in a given state, measurement of a single physical qubit, or just a time step in which a qubit experiences nothing except perhaps some noise (a \emph{wait} location).  Other types of locations are possible, but in any case, we usually assume that every type of location consumes a single time step.  The gates used, states prepared, and measurements made in the fault-tolerant circuit are drawn from some finite universal set.  We also frequently assume that the noise, whatever its form may be, affects all types of locations in more or less the same way.  For instance, we may assume that all locations have the same probability of error.  We often identify in the noisy circuit a subset of locations which have \emph{faults}; locations without faults will behave correctly, according to the appropriate type of location, whereas faulty locations will do something else.

A circuit which we want to perform is broken up into locations, and each location is implemented via a \emph{fault-tolerant gadget} acting on the encoded qubits.  These gadgets must be carefully designed in order to avoid propagating errors badly and to perform the right actions on the encoded data.  A frequent trick is to use \emph{transversal} gates, which are gates which only interact corresponding qubits in different blocks.  In particular, transversal gates never interact two qubits in the same block of the code, so error propagation is less harmful.  However, transversal gates cannot provide a universal gate set~\cite{EK}, so various additional tricks are needed, which frequently involve extra \emph{ancilla} qubits or states.  Between the extra qubits needed to encode data in a QECC and the extra ancilla qubits used in a fault-tolerant protocol, the number of qubits involved in a fault-tolerant circuit can be much greater than the number of qubits that would be needed for an ideal noise-free version of the circuit.  The main focus of this paper is to study the \emph{overhead} of a fault-tolerant protocol, which can be defined as the ratio between the total number of qubits used at once in a fault-tolerant circuit and the number of qubits in the unencoded version of the circuit.

The basic model also contains a number of more significant assumptions about the nature of the circuits used in the protocol and the noise present in the system.  Some of these assumptions are not critical to have a threshold, but are still quite helpful:
\begin{enumerate}
\item \emph{Stochastic noise:} We assume that the set $F$ of faulty locations is a random variable.  That is, with some probability $p_F$, the faults are confined to a set of locations $F$.  Typically, $p_F$ is very small unless $F$ is small.  ($F$ is called a \emph{fault path}.)
\item \emph{No leakage errors:} Errors keep the state within the computational subspace of $n$ qubits.
\end{enumerate}
The assumption of stochastic noise can largely be relaxed, and a threshold exists for noise caused by fairly general interactions with a non-Markovian environment~\cite{AGP06,TB05}.  In this case, the threshold is a bound on the norm of the Hamiltonian coupling between the system and the bath.  I have not carefully checked whether the proof of a threshold for non-Markovian noise still applies to the threshold theorem proved in this paper, but I see no reason that it should not apply.

Leakage errors can be dealt with either by constantly monitoring qubits to identify any that leak out of the computational Hilbert space, or by periodically teleporting data between qubits in order to reset any leakage errors.  Both solutions also work here, provided we apply them at a reduced frequency, perhaps once per syndrome measurement (see Sec.~\ref{sec:FTEC}).

Some assumptions are needed to have a threshold at all:
\begin{enumerate}
\addtocounter{enumi}{2}
\item \emph{Parallel computation:} We assume it is possible to perform in parallel gates on a constant fraction of the qubits in the computer.  Typically, we assume as a simplification that we can perform gates on all qubits at once, provided no qubit is involved in more than one gate at any given time.
\item \emph{Locally decaying noise:} The probability that the fault path contains a specific set of $a$ locations is at most $p^a$.  $p$ is called the \emph{error rate} or \emph{error probability}.
\item \emph{Fresh ancilla qubits:} New physical qubits can be introduced (via a preparation location) at any point during the computation, not just at the beginning of the computation.
\end{enumerate}
A model which has stochastic noise and locally decaying noise, but for which there are no other constraints on the noise, is called a \emph{local stochastic error model}.  An exponential decay of errors is needed in order to prevent many qubits failing all at the same time, an occurrence which can potentially defeat any error-correcting code.  \emph{Independent} noise is when errors are uncorrelated between different qubits.  Independent noise is always locally decaying, but not vice-versa.  An \emph{adversarial} error model is one in which an adversary chooses errors subject to any other constraints in order to cause the most trouble possible.  For instance, an adversarial locally decaying error model would allow the adversary to choose the locations and types of the errors provided the probability that a large group of qubits all have errors is exponentially decaying with the number of qubits.  Adversarial error models are often used to bound error models that have complicated but unspecified types of correlations, and are particularly useful in fault tolerance, because complicated correlations can arise as a consequence of error propagation in a noisy circuit.

It is straightforward to generalize a local stochastic error model to have different error probabilities for different types of locations.  We can assign location type $\ell$ the error rate $p_\ell$.  The probability of having a fault path containing a specific set of $a$ locations, including $a_\ell$ locations of type $\ell$ ($\sum_\ell a_\ell = a$), is at most $\prod_\ell p_\ell^{a_\ell}$.

Parallel computation is needed because otherwise many qubits will need to experience long waits between error correction steps, and errors in the wait locations will overwhelm the capacity of the quantum error-correcting code to correct them.  Fresh ancilla qubits are needed continuously throughout the protocol in order to perform error correction repeatedly; otherwise ancillas will heat up during the course of the computation and will be useless for measuring the errors.  The fresh ancilla qubits can be prepared by hand, or could be created via a suitable noise process that has the effect of cooling some qubits~\cite{BGH13}.

Finally, two of the assumptions in the basic model are critical to the result of this paper even though there is still a threshold (with polylogarithmic overhead) without the assumptions.  In particular, I will assume in this paper that the circuit can use:
\begin{enumerate}
\addtocounter{enumi}{5}
\item \emph{Long-range gates:} That is, the circuit can involve gates interacting arbitrary pairs of qubits, no matter where the two qubits are physically located in the computer.
\item \emph{Fast and reliable classical computation:} Measured qubits produce classical bits, and I assume that we have the capability to perform arbitrary classical computations involving the measurement results with no error and in zero time.  This also implies that measurement locations are available during the computation, not just at the end, and that a measurement location takes a single time step just like any other location.
\end{enumerate}

The efficient LDPC codes I will need involve stabilizer generators interacting qubits which must be placed far apart in any finite-dimensional Euclidean geometry.  In order to quickly measure error syndromes for these codes, we need long-range gates.  It is possible to prove a threshold theorem with gates that only interact nearest-neighbor qubits in one or two dimensions~\cite{AB08,DKLP,Got00}, but this requires codes for which most of the stabilizer generators are physically local.  Codes for which all syndrome bits can be locally measured do not seem capable of encoding qubits without an overhead that grows as the system size grows~\cite{BPT09}, but there may be some possibility of a result comparable to Thm.~\ref{thm:main} if we allow a small number of non-local stabilizer generators.

Free classical computation is a key part of the protocol.  The most common model of classical fault-tolerant computation, where the input must first be encoded, requires logarithmic overhead~\cite{GG94,RS91}.   In a model where the inputs and outputs are provided in some suitable classical error-correcting code, the best known classical fault-tolerant protocol achieves $O(\log n/\log \log n)$ overhead~\cite{Rom06}, and it is unknown if constant overhead is possible.  The goal of this paper is thus to show that we can actually be more efficient with quantum fault tolerance than is known to be possible with classical fault tolerance.  This is not due to any special property of quantum mechanics, but simply because we are willing to outsource a certain amount of the overhead required for fault tolerance to classical bits.  Since classical bits and classical computation are currently much cheaper than qubits and quantum gates, this is a somewhat reasonable assumption, although not all systems will allow it.

Specifically, classical computation is needed to process the error syndromes of codes.  Since an error syndrome for an efficient code will be roughly the size of the code, a non-negligible number of extra bits is needed just to store it.  Processing a large error syndrome takes time, so if we cannot neglect the classical processing time, we will need to continue to measure error syndromes on the code while we are waiting, and that will require still more bits.  However, processing the error syndrome does not require any extra \emph{quantum} bits once the syndrome is measured, and this is why we can get away with a low qubit overhead.

Note, though, that a strict adherence to the rule that classical computation takes zero time means that even exponentially long classical computations can be performed rapidly.  This offers a way to sidestep the quantum computation altogether, by simulating the whole quantum computation on a classical computer.  The simulation will be very inefficient, but that doesn't matter if we neglect the classical resources needed to perform the computation.  A better version of this assumption is for us to neglect only a polynomial number of classical bits and gates (as a function of the size of the quantum computation we wish to perform).  However, in order to state the most general version of Thm.~\ref{thm:main}, I will not make this more sensible version of the classical computation assumption.  Instead, I will just promise not to cheat: I will only use exponential classical computation for syndrome decoding.  That way, we can apply Thm.~\ref{thm:main} to code families whether or not they have efficient classical decoding algorithms.

\section{Main Result}
\label{sec:main}

\begin{thm}
\label{thm:main}
Let $Q_i$ be a family of QECCs with the following properties:
\begin{conditions}
\item $Q_i$ is an $[[n_i,k_i]]$ $(r,c)$-LDPC code, where $r$ and $c$ are constants independent of $i$.
\item As $i \rightarrow \infty$, $n_i \rightarrow \infty$ and $k_i/n_i \rightarrow R$.
\item For some $\beta > 0$, $0 < n_i - n_{i-1} < n_{i-1}^\beta$.
\label{item:slowgrowth}
\item Suppose the code $Q_i$ experiences local stochastic noise.  It is initialized with errors with error rate $\tilde{p}$ per qubit, has further errors during error correction with error rate $p$ per physical qubit per syndrome extraction, and each bit of the error syndrome is flipped with error rate $q$.  Assume error correction is done by measuring the error syndrome a polynomial number of times $T(n_i)$, followed by classical syndrome decoding which has depth $h(n_i)$.  Then there exists a decoding algorithm, thresholds $p_0$, $p_1$, and $p_2$, and function $g(n)$ such that if $\tilde{p} < p_0$, $p < p_1$, and  $q < p_2$, then the probability of logical error after syndrome decoding is $O(1/g(n_i))$ as $i \rightarrow \infty$ and, if there is no logical error, the errors in the physical codeword after the end of error correction can also be described by a local stochastic noise model with error rate at most $p_0/3$.  I will assume the functions $T(n)$, $g(n)$, and $h(n)$ are non-decreasing as $n$ increases.
\label{item:adversarial}
\end{conditions}
Choose $\alpha<1$ with $0 < \alpha \beta < 1$.  Then, for all $\eta > 1$, all $\epsilon > 0$, and all polynomials $f(n) = o(g(n^\alpha))$, there exists a threshold error rate $p_T(\eta)$ and a threshold size $k_0 (\eta,f,\epsilon)$ such that, for any sequential logical quantum circuit $\circuit{C}$ using $k$ qubits and $f(k)$ locations, with $k> k_0$, there exists a fault-tolerant simulation $\circuit{\tilde{C}}$ of $\circuit{C}$ which uses at most $\eta k/R$ physical qubits.  In the basic model of fault tolerance with adversarial local stochastic noise with error rate $p < p_T$, $\circuit{\tilde{C}}$ outputs a distribution which has statistical distance at most $\epsilon$ from the output distribution of $\circuit{C}$.

Furthermore, the fault-tolerant simulation uses $O( f(k) T(2k^{\alpha'}) k^{\alpha'}/(\eta - 1) + f(k) k^{\alpha'} \polylog (k/\epsilon))$ total physical locations and
classical computation with depth $O(h(2k^{\alpha'}) + \log \log(k/\epsilon))$ per logical time step, with $\alpha' = \max(\alpha, \alpha \beta)$.  In particular, if $T(n_i)$ is polynomial in $n_i$, then the total number of locations in the circuit is polynomial in $k$, and if $Q_i$ has a polynomial-time decoding algorithm, then the fault-tolerant simulation uses only a polynomial amount of classical computation as well.
\end{thm}

We know that codes exist satisfying all of the conditions of Thm.~\ref{thm:main}.  The known families of codes and their advantages and drawbacks will be discussed in Sec.~\ref{sec:LDPC}.

In other words, Thm.~\ref{thm:main} says that if we are below the threshold, arbitrarily long reliable quantum computations are possible with asymptotic overhead equal to that of the QECC being used.  However, there are two substantial differences from the usual threshold theorem besides the reduced overhead.  First, there is a minimum size $k_0$ for the circuit $\circuit{C}$.  This appears in the theorem because the efficient codes we need may only exist for large block sizes, and even when they exist, may not have low logical error rates unless the block sizes are large.  

Second, there is a bound $f(n)$ on the size of the circuit $\circuit{C}$.  This is needed for a more subtle reason.  To achieve a final error rate of $\epsilon$, the error rate per logical location must be about $\epsilon/f(n)$.    That is, the required logical error rate decreases as the size of $\circuit{C}$ increases.  We must deal with this by increasing the size of the code so that it is better at correcting errors.  In the usual threshold theorem, the increase in code size is absorbed into the extra overhead needed for a long computation.   In Thm.~\ref{thm:main}, we need to make sure that using a larger code does not require more overhead.  We pick an $i$ so that code $Q_i$ has logical error rate below $\epsilon/f(k)$.  We must then be sure that the code $Q_i$ is of an appropriate size for the number of logical qubits in $\circuit{C}$.  If $f(n)$ grows too fast, faster than $g(n)$, then we have to pick an extra-large code in order for it to be sufficiently good at correcting errors, but such a big code has extra overhead.  This is discussed in detail in Sec.~\ref{sec:gates}.

For the best code families, $g(n) = \exp (\poly(n))$, so the theorem applies for arbitrary polynomial-size circuits.  However, there are interesting code families for which $g(n) = \poly(n)$ only, limiting the scaling of size of the logical computation.  Also note that, while the minimum circuit size $k_0$ depends on $f$, normally the threshold $p_T$ does not.  However, the threshold \emph{does} depend on the rate of the codes used and how close we want to get to the optimal overhead.

The first two conditions on the code family are straightforward: They simply say that there is a family of LDPC codes with asymptotic rate $R$.  Condition~(\ref{item:slowgrowth}) is needed to be sure that there is a code in the family of about the right size for the computation.  Condition~(\ref{item:adversarial}) says that the code family can correct errors at a constant rate.  It is important that the error correction procedure be robust --- it must still work when there are faults in the error correction procedure and errors in the measured syndrome.  
For the purpose of condition~(\ref{item:adversarial}), I assume a simplified model of fault tolerance where the state is initialized with errors at rate $\tilde{p}$ and then the syndrome is measured repeatedly in some fault-tolerant way.  Each time the syndrome is measured, errors are added to the system at rate $p$.  Then the syndrome is extracted perfectly, but before it is reported, syndrome bits are flipped with error rate $q$.  There can be correlations between the different errors provided the overall distribution is governed by a local stochastic noise model.

\section{Efficient LDPC Codes}
\label{sec:LDPC}

A number of different constructions are known giving families of $[[n,k,d]]$ LDPC codes, and a few of them are able to achieve a constant encoding rate $k/n$.  However, none of the known families are \emph{good} codes, which would require that both $k/n$ and the relative distance $d/n$ are non-zero constants as $n \rightarrow \infty$.  The best known performance (from hypergraph product codes~\cite{TZ}) is constant $k/n$ and $d = O(\sqrt{n})$.

When there are $n$ physical qubits and errors are occurring with probability $p$ per qubit, we expect there will be about $pn$ errors.  For large $n$, it thus seems like a distance $O(\sqrt{n})$ code will be insufficient.  However, Kovalev and Pryadko~\cite{KP12b} showed that while there may be certain particularly bad errors of length about $\sqrt{n}$ that cannot be corrected, \emph{typical} depolarizing errors at error rate $p$ can be corrected by an LDPC code, provided $p$ is less than some threshold error rate $p_0$.  They prove this by showing that for low error rates, the errors tend to form small clusters in a low-degree graph determined by the stabilizer generators.  

However, the result of \cite{KP12b} applies to independent errors, but condition~(\ref{item:adversarial}) of Thm.~\ref{thm:main} asks for decoding under adversarial noise.  Therefore, the first order of business is to generalize the result of \cite{KP12b} to adversarial stochastic noise.

In particular, Lemma~2 of \cite{KP12b} is essential to the proof in that paper.  The lemma states that for sufficiently low error rates, the probability of large clusters of errors is exponentially small.  Unfortunately, the proof in \cite{KP12b} makes explicit use of the independence of errors.  To produce the same result for adversarial stochastic noise, we need a more careful counting of the number of clusters of a given size in a bounded-degree graph.  An appropriate bound is provided by Lemma~5 of \cite{AGP08}:
\begin{lem}
Consider a specific set $S$ consisting of $t$ nodes in a graph for which every node has degree at most $z$.   Let $M_z (s, S)$ be the number of sets containing $S$ and a total of $s$ nodes (i.e., $s-t$ nodes beyond those in $S$), and which are a union of connected clusters, each of which contains a node in $S$.  Then  $M_z(s,S) \leq e^{t-1} (ze)^{s-t}$, with $e$ the usual base of the natural logarithm.
\label{lem:numclusters}
\end{lem}
In \cite{AGP08}, the lemma is stated to provide a bound on the number of connected sets of size $s$ containing $S$, but the proof also provides a bound (as stated above) on the number of collections of clusters, each of which contains an element of $S$.

The graph to which we wish to apply this lemma is a graph defined by the stabilizer of the code.  The graph has one node for each qubit in the code, and two nodes are connected by an edge in the graph iff there is an element of the stabilizer acting non-trivially on both qubits.  I will refer to this graph as the \emph{adjacency graph} of the code.  As noted in \cite{KP12b}, for an $(r,c)$-LDPC code, this graph has bounded degree, with degree at most $z = (r-1)c$ --- each generator containing a qubit connects it to $r-1$ other qubits, and there are at most $c$ generators containing the qubit.

We can then prove a generalization of Thm.~3 of \cite{KP12b}:
\begin{thm}
Let $Q_i$ be a family of quantum $(r,c)$-LDPC codes, and let $Q_i$ be an $[[n_i,k_i,d_i]]$ code with $n_i, d_i \rightarrow \infty$ for $i \rightarrow \infty$.  Then there exists a threshold $p_0$ such that if the code experiences local stochastic noise with error probability $p < p_0$, then there exists a (not necessarily efficient) decoding algorithm such that the logical error rate approaches $0$ proportionally to $n_i (p/p_0)^{d_i/2}$ as $i \rightarrow \infty$.
\label{thm:adversarialLDPC}
\end{thm}
In particular, if $d_i \propto \sqrt{n_i}$, the logical error rate goes to $0$ superpolynomially.

\begin{proof}
I will largely follow the proof of \cite{KP12b}, making only a few changes to deal with adversarial errors instead of independent errors.  The intuition is that with high probability, errors will be contained in small separated clusters, which can be corrected independently.

The decoding algorithm is simply to choose the lowest-weight error consistent with the error syndrome.  If there is more than one, choose one at random.  (This may not be the same as maximum-likelihood decoding because for a degenerate code, maximum-likelihood decoding should add up the total probability of equivalent errors.)

I claim that if this decoding procedure results in a logical error, then there must be some connected cluster of $s \geq d_i$ qubits such that the cluster contains $m \geq \lceil s/2 \rceil$ errors in the actual fault path.  To prove the claim, assume there is a logical error after decoding.  Note that the deduced (incorrect) error $E$ and the actual error $F$ must have the same error syndrome and that the product $G = EF$ therefore has trivial error syndrome.  Consider a minimal connected cluster $S$ for which $G$ has no support on qubits on the border of $S$, defined as the set of qubits which are adjacent to a qubit in $S$ but are not themselves in $S$.  $G|_S$ (the Pauli operator $G$ restricted to just qubits in $S$) must have trivial error syndrome as well. $G|_S$ must act non-trivially on every qubit within $S$, because otherwise there is a smaller $S$ with the desired property, produced by eliminating qubits on which $G$ acts trivially.  (If removing a qubit disconnects the cluster, just take one of the connected components of the remainder.)  There must be some such cluster $S$ such that $G|_S$ is not an element of the stabilizer --- otherwise $G$ itself would be in the stabilizer and the error decoding would have succeeded.  Let the size of $S$ be $s$.  If $G|_S$ is not element of the stabilizer, then $s \geq d_i$.  Furthermore, we know that $\wt E|_S \leq \wt F|_S$ or else we could modify $E$ by making it equal to $F$ within $S$ and in the process shorten $E$ without changing the error syndrome.  Therefore, the cluster $S$ contains $m \geq s/2$ qubits with errors.

For any given cluster $S$ of size $s$, the probability of having at least $\lceil s/2 \rceil$ errors in $S$ is at most $2^s p^{s/2}$.  (There are less than $2^s$ subsets of large enough size, and the probability of having errors on every qubit in a particular subset is at most $p^{s/2}$ by the definition of the local stochastic noise model.)  By Lemma~\ref{lem:numclusters}, the number of clusters of size $s$ containing a particular qubit is at most $(ze)^{s-1}$, and therefore the total number of clusters of size $s$ is at most $n_i (ze)^{s-1}$.  The total probability of having a cluster of size $s\geq d_i$ for which more than half the qubits have errors is therefore
\begin{equation}
\mathrm{Prob(error)} \leq \frac{n_i}{ze} \sum_{s \geq d_i} (2 ze \sqrt{p})^s = \frac{n_i}{ze} \frac{(2 ze \sqrt{p})^{d_i}}{1-2 ze \sqrt{p}}
\end{equation}
if $2ze \sqrt{p} < 1$, in which case this goes to $0$ as $d_i \rightarrow \infty$.  When $2ze \sqrt{p} \geq 1$, we get no meaningful bound on the error probability.  Let $p_0 = (2ze)^{-2}$.  Then for $p < p_0$,
\begin{equation}
\mathrm{Prob(error)} = \frac{n_i}{ze(1-2 ze \sqrt{p})} (p/p_0)^{d_i/2}.
\end{equation}

\end{proof}

Condition (\ref{item:adversarial}) of Thm.~\ref{thm:main} goes further than Thm.~\ref{thm:adversarialLDPC}, also requiring that the code be able to correct for faulty syndrome measurements and for new errors occurring during error correction.  Kovalev and Pryadko~\cite{KP12b} sketch an argument that LDPC codes can also correct for syndrome measurements with independent stochastic errors.  The same argument applies to adversarial local stochastic noise, so I recap their construction here.  In fact, condition (\ref{item:adversarial}) of Thm.~\ref{thm:main} asks for an additional strengthening beyond that in the form of a constraint on the residual physical errors after error correction.  Both results are covered by the following theorem:

\begin{thm}
Let $\{Q_i\}$ be a family of quantum $(r,c)$-LDPC codes, and let $Q_i$ be an $[[n_i,k_i,d_i]]$ code with $n_i, d_i \rightarrow \infty$ for $i \rightarrow \infty$.  Suppose the code suffers from local stochastic noise with error probability $\tilde{p}$ per qubit initially, error rate $p$ per physical qubit per syndrome extraction, and error rate $q$ for each measured syndrome bit.  Assume the syndrome measurement is repeated $T = d_i$ times.  Then there exist thresholds $p_0$, $p_1$, and $p_2$ such that if $\tilde{p} < p_0$, $p < p_1$, $q <p_2$, then there exists a (not necessarily efficient) decoding algorithm such that the logical error rate approaches $0$ as $i \rightarrow \infty$ proportionally to $n_i (p'/p_0)^{d_i/4}$ or $n_i T (p'/p_0)^{d_i/2}$, with $p' = \max(p,q)$, or as $n_i (\tilde{p}/p_0)^{d_i/2}$ (whichever is largest).  Furthermore, if there is no logical error, the state exiting the final syndrome measurement can be described as having errors generated by a local stochastic noise model with error probability at most $p_0/3$ per qubit.
\label{thm:adversarialFT}
\end{thm}

As with condition~(\ref{item:adversarial}) of Thm.~\ref{thm:main}, Thm.~\ref{thm:adversarialFT} is for a simplified model of fault tolerance where an error during the syndrome extraction procedure either affects a single syndrome bit  and no data qubits or affects a physical data qubit and then is taken into account in the next syndrome measurement.

\begin{proof}
We decode by finding the lowest-weight fault path consistent with the (faulty) syndrome measurements recorded.  For the purposes of this paper, I will count weight by simply adding up the number of syndrome bit errors and the number of physical errors at all times.  A more sophisticated approach would weight initialization errors, new errors on physical qubits in the QECC, and syndrome bit errors differently to account for their different probabilities.

Let $F_t$ be the actual set of new physical errors occurring at time $t$ (represented as a Pauli operator), let $E_t$ be the deduced set of new qubit errors at time $t$, let $B_t$ be the actual set of syndrome bit errors at time $t$ (represented as a binary vector), and let $C_t$ be the deduced set of syndrome bit errors at time $t$. We count an error as occurring at time $t$ if it happens in a way that lets it appear in the syndrome measurement at time $t$.   Time $t=1$ has the initialization errors as well as the errors associated with the first error syndrome measurement.  

The observed change in the syndrome from time $t-1$ to time $t$ is $\Delta_t$, which can be broken up into the syndrome bit errors and the new physical errors.  Since the actual and deduced errors must both be consistent with $\Delta_t$, we have for $t>1$
\begin{equation}
B_{t-1} \oplus B_{t} \oplus \sigma(F_t) = \Delta_t = C_{t-1} \oplus C_t \oplus \sigma(E_t),
\end{equation}
where $\sigma(P)$ is the function that gives the error syndrome of the Pauli $P$.  $\sigma$ is a homomorphism, so $\sigma (E_t F_t) = \sigma(E_t) \oplus \sigma(F_t) = B_{t-1} \oplus B_t \oplus C_{t-1} \oplus C_t$.  We take $\Delta_1$ to be the measured error syndrome at time $1$, so
\begin{equation}
B_1 \oplus \sigma(F_1) = \Delta_1 = C_1 \oplus \sigma(E_1).
\end{equation}

The overall effect of errors plus corrections over the full cycle of $T$ repetitions of the syndrome measurement is $\prod_t (E_t F_t)$.  This has the following error syndrome:
\begin{align}
\sigma \left(\prod_{t=1}^{T} (E_t F_t) \right) =& \bigoplus_{t=1}^T \sigma(E_t F_t) \\
=& B_1 \oplus C_1 \oplus \bigoplus_{t=2}^T (B_{t-1} \oplus B_t \oplus C_{t-1} \oplus C_t) \\
=& B_T \oplus C_T,
\end{align}
the last equality following because it is a telescoping sum.  If after the final measurement we were able to perform a perfect syndrome measurement, we would get $B_T \oplus C_T$ and deduce $G$, the lowest-weight physical error with that syndrome.  We can therefore say that the logical state of the code has been altered during the course of error correction by the operator $\prod_t (E_t F_t) G$.  This must be a logical operator of the code, and if it is non-trivial, we can say that decoding has failed.

We represent this procedure with a new graph, which I will call the \emph{syndrome adjacency graph}.  Begin with the adjacency graph of the LDPC code and replicate it once for each time from $1$ to $T+1$.  We can label the nodes of the replicated graph by $(x,t)$, where $x$ is a node of the original graph and $t$ is a time, which must be an integer.  Add a new node for each pair (syndrome bit, time step), for $t = 1, \ldots, T$.  The new node labelled $(b,t)$ is connected to $(x,t)$ and $(x,t+1)$ for all qubits $x$ in the support of generator $b$ of the stabilizer.  That is, $(b,t)$ is connected to all the qubits measured by that syndrome bit at both the current time and the following time.  
% $(b,t)$ is also connected to $(b,t+1)$ and $(b,t-1)$.  
The nodes $(x,t)$ in the syndrome adjacency graph have degree at most $z' = z + 2c$ and the nodes $(b,t)$ have degree at most $2r < z'$.  Figure~\ref{fig:FTsyndrome} gives an example of such a graph.

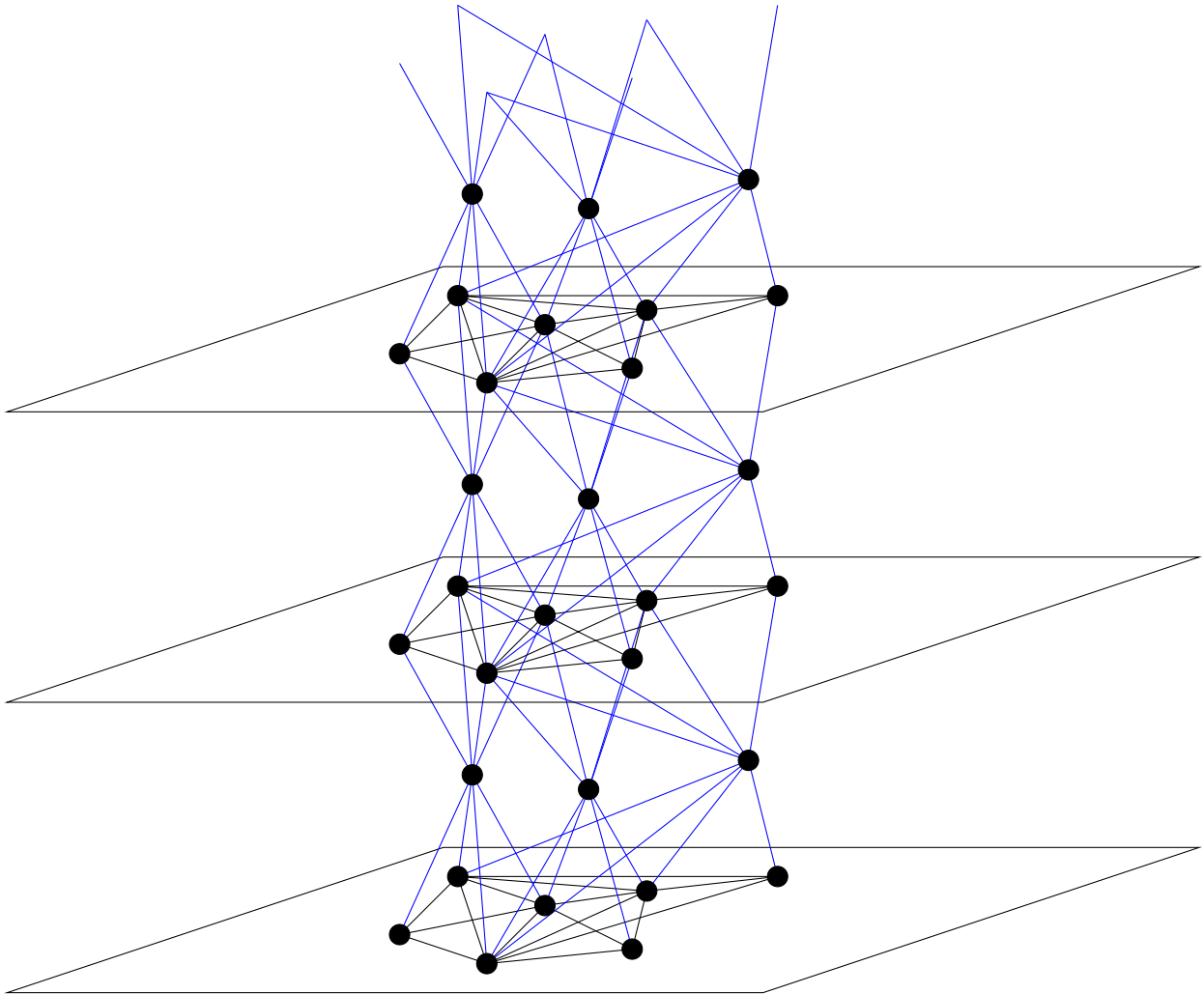
\begin{figure}
\begin{center}
\begin{tikzpicture}[scale=2]

%\draw [red] (0.4,0) ++(3,1.5) -- ++(0,5.3) ++(0.8,-0.1) -- ++(0,-5.3) ++(1.1,0.2) -- ++(0,5.3);

\foreach \y in {0, 2, 4} {
	\draw (0.2,\y) -- ++(5.2,0) -- ++(3,1) -- ++(-5.2,0) -- cycle;

	\draw (0.4,\y) ++(2.5,0.4) -- ++(0.6,-0.2) -- ++(0.4,0.4) -- ++(-0.6,0.2) -- ++(-0.4,-0.4);
	\draw (0.4,\y) ++(2.5,0.4) -- ++(1,0.2) ++(-0.6,0.2) -- ++(0.2,-0.6);

	\draw (0.4,\y) ++(3.1,0.2) -- ++(1,0.1) -- ++(0.1,0.4) -- ++(-0.7,-0.1);
	\draw (0.4,\y) ++(3.5,0.6) -- ++(0.6,-0.3) ++(-1,-0.1) -- ++(1.1,0.5);

	\draw (0.4,\y) ++(3.1,0.2) -- ++(2,0.6);
	\draw (0.4,\y) ++(5.1,0.8) -- ++(-0.9,-0.1) -- ++(-1.3,0.1) -- cycle;

	\draw [blue] (0.4,\y) ++(2.5,0.4) -- ++(0.5,1.1) -- ++(-0.5,0.9);
	\draw [blue] (0.4,\y) ++(2.5,0.4) ++(0.6,-0.2) -- ++(-0.1,1.3) -- ++(0.1,0.7);
	\draw [blue] (0.4,\y) ++(2.5,0.4) ++(0.6,-0.2) ++(0.4,0.4) -- ++(-0.5,0.9) -- ++(0.5,1.1);
	\draw [blue] (0.4,\y) ++(2.5,0.4) ++(0.4,0.4) -- ++(0.1,0.7) -- ++(-0.1,1.3);

	\draw [blue] (0.4,\y) ++(3.1,0.2) -- ++(0.7,1.2) -- ++(-0.7,0.8);
	\draw [blue] (0.4,\y) ++(3.1,0.2) ++(1,0.1) -- ++(-0.3,1.1) -- ++(0.3,0.9);
	\draw [blue] (0.4,\y) ++(3.1,0.2) ++(1,0.1) ++(0.1,0.4) -- ++(-0.4,0.7) -- ++(0.4,1.3);
	\draw [blue] (0.4,\y) ++(3.1,0.2) ++(0.4,0.4) -- ++(0.3,0.8) -- ++(-0.3,1.2);

	\draw [blue] (0.4,\y) ++(3.1,0.2) -- ++(1.8,1.4) -- ++(-1.8,0.6);
	\draw [blue] (0.4,\y) ++(3.1,0.2) ++(2,0.6) -- ++(-0.2,0.8) -- ++(0.2,1.2);
	\draw [blue] (0.4,\y) ++(3.1,0.2) ++(2,0.6) ++(-0.9,-0.1) -- ++(0.7,0.9) -- ++(-0.7,1.1);
	\draw [blue] (0.4,\y) ++(3.1,0.2) ++(2,0.6) ++(-0.9,-0.1) ++(-1.3,0.1) -- ++(2,0.8) -- ++(-2,1.2);

	\draw [fill=black] (0.4,\y) ++(2.5,0.4) circle (0.07cm) ++(0.6,-0.2) circle (0.07cm) ++(0.4,0.4) circle (0.07cm) ++(-0.6,0.2) circle (0.07cm);
	\draw [fill=black] (0.4,\y) ++(3.1,0.2) ++(1,0.1) circle (0.07cm) ++(0.1,0.4) circle (0.07cm);
	\draw [fill=black] (0.4,\y) ++(3.1,0.2) ++(2,0.6) circle (0.07cm);

	\draw [fill=black] (0.4,\y) ++(3,1.5) circle (0.07cm) ++(0.8,-0.1) circle (0.07cm) ++(1.1,0.2) circle (0.07cm);
	}

\end{tikzpicture}
\caption{An example of a syndrome adjacency graph.  Layers of the original code adjacency graph alternate with layers representing the syndrome bits.  Different layers represent different times.  The layers surrounded with parallelograms represent copies of the original adjacency graph; the other layers represent syndrome bits.  Blue edges are the new ones not present in the original graph.}
\label{fig:FTsyndrome}
\end{center}
\end{figure}

Now imagine marking a node $(x,t)$ on the syndrome adjacency graph if $E_t F_t$ has an error on qubit $x$, and mark a node $(b,t)$ if bit $b$ of $B_t \oplus C_t$ is $1$.  We include a time step $T+1$ with a copy of the original graph, connected to syndrome bits at time $T$ as above, and mark nodes $(x,T+1)$ if the residual error $G$ has an error on qubit $x$.

The marked nodes form clusters in the syndrome adjacency graph.  Suppose a cluster $K$ is completely contained within the interval $[t_1, t_2]$ with $1 \leq t_1 \leq t_2 \leq T+1$.  Then we have that $\sigma (\prod_{t=t_1}^{t_2} (E_t|_K F_t|_K) G|_K) = 0$, where $E_t|_K$ and $F_t|_K$ consist of just those qubit Pauli tensor factors that are in the cluster $K$ at time $t$, and $G|_K$ just those qubit Pauli tensor factors that are in the cluster at time $T+1$.  ($G|_K = I$ if $t_2 < T+1$.)  Thus, we can consider $P_K = \prod_{t=t_1}^{t_2} (E_t|_K F_t|_K) G|_K $ to be the logical error due to that cluster; $P_K$ has trivial error syndrome for the code. 

I will assume the overall procedure fails if either there is a connected cluster stretching from a node of the form $(x,1)$ to a node of the form $(x,T+1)$ or if the decoding algorithm results in a logical error.  As in the proof of Thm.~\ref{thm:adversarialLDPC}, 
\begin{equation}
\sum_{t=t_1}^{t_2} \wt (E_t|_K) + \wt(C_t|_K) \leq \sum_{t=t_1}^{t_2} \wt (F_t|_K) + \wt(B_t|_K),
\end{equation}
and if the decoding procedure fails, there must be some cluster $K$ for which $P_K$ is a non-trivial logical error.  $B_t|_K$ and $C_t|_K$ are the syndrome error sets restricted to just those syndrome bits present in the cluster $K$ at time $t$.  Suppose the total size of the cluster is $s$.  If $t_2 < T+1$, it must therefore be the case that the total number of errors in the cluster (syndrome bit errors plus qubit errors)  is at least $s/2$.  For clusters with $t_2 = T+1$, note that $\sigma(G|_K) = B_T|_K \oplus C_T|_K$.  Since $G|_K$ is the smallest error consistent with the syndrome $B_T|_K \oplus C_T|_K$, which is also the syndrome of $\prod_{t=t_1}^{t_2} (E_t|_K F_t|_K)$, it must be the case that $\wt (G|_K) \leq \sum_{t=t_1}^{t_2} (\wt E_t|_K + \wt F_t|_K)$.  This means $\wt G|_K \leq s/2$ and the total weight of the $E_t|_K$, $F_t|_K$, $B_t|_K$, and $C_t|_K$ is at least $s/2$.  Since at least half of those are real errors, we can conclude that the total number of actual errors in the cluster is at least $s/4$ even when the cluster ends at time $T+1$.

Since $P_K$ is a logical error, $\wt P_K \geq d_i$, and since $P_K$ is derived from the cluster, $\wt P_K \leq s$.  We treat separately clusters which have $(E_1 F_1)|_K = I$ and $t_2 < T+1$, clusters which have $t_2 = T+1$ but $(E_1 F_1)|_K = I$, clusters which have $t_2 < T+1$ and $(E_1 F_1)|_K \neq I$, and clusters which have $(E_1 F_1)|_K \neq I$ and $t_2 = T+1$.

First I consider the simplest case, clusters with $(E_1 F_1)|_K = I$ and $t_2 < T+1$.  The probability of having a logical error due to such a cluster is bounded by the probability of having in the new graph a cluster of size $s \geq d_i$ containing at least $s/2$ errors.  In this case, we can simply duplicate the calculation in the proof of Thm.~\ref{thm:adversarialLDPC}, replacing $z$ with $z'$, $p$ with $p'= \max(p,q)$, and $n_i$ with $n' = n_i (T-1) + (n_i - k_i) T$, the number of nodes in the graph excluding qubits at times $1$ and $T+1$.  Because the cluster does not contain any of the initial errors, it is sufficient to consider an error probability $p'$ which is the larger of $p$ and $q$.  I will omit the detailed calculation since the next two cases will almost always be more stringent unless $T$ is very large.

The probability of having a logical error due to a cluster which ends at time $T+1$ but has no initialization errors can be calculated in the same way, but with $s/2$ replaced by $s/4$.  We can use $n_i$ in this case, since the cluster must contain a qubit at the final time.  We find
\begin{equation}
\mathrm{Prob(error)} \leq \frac{n_i}{z'e} \sum_{s \geq d_i} (2 z'e p'^{1/4})^s = \frac{n_i}{z'e} \frac{(2 z'e p'^{1/4})^{d_i}}{1-2 z'e p'^{1/4}}.
\end{equation}
We get a threshold of $p_f = (2z'e)^{-4}$, and for $p' < p_f$,
\begin{equation}
\mathrm{Prob(error)} = \frac{n_i}{z'e(1-2 z'e p'^{1/4})} (p'/p_f)^{d_i/4}.
\end{equation}

For a cluster which contains initialization errors, but ends before time $T+1$, we must use error probability $p'' = \max(\tilde{p},p,q)$, but can use both $s/2$ errors in the clusters and $n_i$ specific single qubits contained in the cluster at time $t=1$.  We have
\begin{equation}
\mathrm{Prob(error)} \leq \frac{n_i}{z'e} \sum_{s \geq d_i} (2 z'e \sqrt{p''})^s = \frac{n_i}{z'e} \frac{(2 z'e \sqrt{p''})^{d_i}}{1-2 z'e \sqrt{p''}}.
\end{equation}
We get a threshold of $p_i = (2z'e)^{-2}$, and for $p'' < p_i$,
\begin{equation}
\mathrm{Prob(error)} = \frac{n_i}{z'e(1-2 z'e \sqrt{p''})} (p''/p_i)^{d_i/2}.
\end{equation}

The last case is a cluster which reaches from time $t_1 = 1$ with initialization errors to $t_2 = T+1$.  In order to do this, the cluster must be size $s \geq 2T+1$ because there must be at least one marked node of the form $(x,t)$ and one of the form $(b,t)$ for each time $t$ between $1$ and $T$, plus at least one node $(x,T+1)$.  Such a cluster does not have to cause a logical error to give us a problem (it may instead result in too many physical errors surviving through the error correction cycle), but the other arguments above still apply.  We therefore find that the cluster must contain at least $s/4$ actual errors.  We repeat the calculation using error rate $p''$ (since we include the initial time in the cluster), and find that 
\begin{equation}
\mathrm{Prob(error)} = \frac{n_i (p''/p_f)^{1/4}}{z'e(1-2 z'e p''^{1/4})} (p''/p_f)^{T/2},
\end{equation} 
giving us a more stringent requirement $p'' < p_f$.

The state exiting the error correction has the error $G$ on it.  We wish to bound the probability that $G$ has errors on a specific set $S$ of $a$ qubits.  Again, $G$ is the smallest-weight error with syndrome $B_T \oplus C_T$.  By breaking down into clusters $K$, we have from before that $\wt G|_K \leq \sum \wt E_t|_K + \wt F_t|_K$, with the sum taken over times for the cluster.

Let us sum up the total size of clusters (built from locations in $E_t F_t$ and $B_t \oplus C_t$, as before) attached to qubits in $S$ (at time $T+1$); call this total size $s$.  The number of groups of clusters of size $s$ in the graph containing the $a$ qubits in the set $S$ is at most $e^{a-1} (z'e)^{s-a}$ by lemma~\ref{lem:numclusters}. The argument from the previous paragraph tells us that $s \geq 2a$, and the argument from earlier in the proof tells us that these clusters contain at least $s/4$ actual errors.  Because we are assuming that error correction succeeded, there is no cluster reaching from the initialization stage to time $T+1$, so we can use error rate $p'$.  The probability of having clusters which contain all the qubits in the set $S$ is thus at most
\begin{equation}
\sum_{s \geq 2a} e^{a-1} (z'e)^{s-a} 2^s p'^{s/4} = \frac{1}{ez'^a} \frac{(2z'e p'^{1/4})^{2a}}{1-2z'e p'^{1/4}} = \frac{1}{e(1 - 2z'e p'^{1/4})} (4z' e^2 \sqrt{p'})^a.
\end{equation}
Provided $e(1 - 2z'e p'^{1/4}) > 1$ and $4z' e^2 \sqrt{p'} < p_0/3$, the output error rate will have the desired properties.

Putting all these requirements together, we set 
\begin{align}
p_0 &= p_f = (2z'e)^{-4} \\
p_1 = p_2 &= p_f^2 / (144 (z')^2 e^4) = (192 (z')^5 e^6)^{-2}.  
\end{align}
One can then check that if $\tilde{p} < p_0$, $p < p_1$, and $q< p_2$, then all the required threshold conditions are satisfied.

\end{proof}

The thresholds implied by the proof of this theorem are extremely bad, but there are many simplifications involved in the proof which suggest the reality will not be nearly as bad as this.  I have taken the worst possible error rate to affect the whole cluster, which is not going to be the case.  The fourth power in $p_0$ and the $12$th powers that appear in $p_1$ and $p_2$ are a consequence of poor control of the errors after the termination of the error correction cycle.  It is actually quite unlikely that a large cluster would result in a number of surviving qubit errors equal to half the size of the cluster.  This approximation results in what is probably a superfluous factor of $2$ in the exponents of $p_0$ and the total probability of failure, and an extra factor of $4$ or more in the exponent of $p_1$ and $p_2$.

Furthermore, it is not a good idea to simply stop error correction in order to perform gates.  The theorem assumes that information about error syndromes will propagate within a single error correction cycle containing $T$ syndrome extractions, but that no information will be available at the start of the next cycle.  In fact, we can track error propagation through the gate(s) performed between error correction cycles and use that information to improve our syndrome decoding.  Ultimately, this strategy should result in a failure rate which behaves like the most favorable case, due to clusters which do not reach either the starting time or ending time.

There are two known methods of constructing codes which satisfy all the conditions of Thm.~\ref{thm:main}.  Tillich and Zemor~\cite{TZ} recently gave a code construction (the ``hypergraph product codes'') which takes a classical code $C$ with parameters $[n,k,d]$ and produces a quantum code $Q$ with parameters $[[n',k',d']] = [[n^2 + (n-k)^2, k^2, d]]$.  If $C$ is a classical $(r,c)$-LDPC code, then $Q$ is a quantum $(r', c')$-LDPC code, with $r' \leq r+c$, $c' \leq \max(r,c)$.  These codes have the highest known distance of LDPC codes with constant rate, but have a big disadvantage in that we do not know how to efficiently decode them.  The other method is to use homology constructions based on the toric code on hyperbolic manifolds.  This method gives codes with lower distance but with good decoding algorithms.  Freedman, Meyer, and Luo~\cite{FML} showed that there exists a family of $2$-dimensional hyperbolic manifolds such that the surface codes on these manifolds have parameters $[[n, k, O(\sqrt{n/k} \log k)]]$ and are $(7,4)$-LDPC codes.  A similar construction with $4$-dimensional hyperbolic manifolds gives LDPC codes which actually have distance $O(n^{\epsilon})$ for some $\epsilon > 0$~\cite{GL}. In this case, $\epsilon$ is known to  be at most $0.3$, so these codes have lower distance than the hypergraph product codes.  The details of these constructions do not particularly matter for proving Thm.~\ref{thm:main}, but it is worth noting that they all produce CSS codes.

To use the hypergraph product codes in Thm.~\ref{thm:main}, choose a desired classical asymptotic rate $R_c$.  Then there is a suitable choice of constants $(r,c)$ such that a random $[n,k,d]$ classical $(r,c)$-LDPC code, for at least a constant fraction of (large) $n$, has, with high probability, asymptotic rate $k/n \rightarrow R_c$ and a constant relative distance $d/n$~\cite{Gal}.  Applying the hypergraph product construction~\cite{TZ} gives a family of $[[n_i,k_i,d_i]]$ quantum LDPC codes satisfying all the conditions required for Thm.~\ref{thm:main}.  They have distance $d_i = O(\sqrt{n_i})$, any $\beta > 1/2$, and asymptotic rate $R = R_c^2/[1+(1-R_c)^2]$.  Thm.~\ref{thm:adversarialFT} tells us these codes satisfy condition~(\ref{item:adversarial}) of Thm.~\ref{thm:main} with $g(n_i) = \exp(-O(\sqrt{n_i}))$.  Explicit classical constructions based on expander codes~\cite{SS96} also work.  The hypergraph product codes were improved in~\cite{KP12a}, but the asymptotic behavior of the improved construction is similar.  Unfortunately, no efficient decoding algorithm is known for hypergraph product codes.

Note that it is possible to choose the rate $R$ to be arbitrarily close to $1$ for hypergraph product codes.  Naturally, this has the disadvantage of lowering the thresholds $p_0$, $p_1$, and $p_2$, but it still might be an attractive choice if it is possible to construct high-cost but very high-fidelity qubits.

The two-dimensional hyperbolic codes~\cite{FML} satisfy all the constraints of Thm.~\ref{thm:main}, but only marginally so.  For constant $R=k/n$, the distance is only $O(\log n)$.  \cite{FML} does not spell out for which parameters codes exist, but it appears to be polynomial, giving $\beta < 1$.  Applying Thm.~\ref{thm:adversarialFT}, we find that the error suppression is $1/g(n) = (Cn \log n) (p/p_0)^{O(\log n)} = (Cn \log n) n^{-O(\log (p_0/p))}$.  That is, errors are suppressed only polynomially with $n$.  If $p$ is close to $p_0$, the polynomial suppression in the second term may be overwhelmed by the $n\log n$ factor in the first term.  However, by ensuring that $p$ is sufficiently small compared to $p_0$, we can not only ensure that the second term dominates and errors are suppressed asymptotically with $n$, but that the suppression is at any desired polynomial rate.  In other words, by choosing the thresholds $p_0$, $p_1$, and $p_2$ for condition~(\ref{item:adversarial}) of Thm.~\ref{thm:main} to be below the thresholds that come directly out of Thm.~\ref{thm:adversarialFT}, we can make the error suppression function $g(n) = \Theta(n^\gamma)$ for any desired constant $\gamma$.  However, note that the actual choice of $p_0$, $p_1$, and $p_2$ depend on $\gamma$. We can choose any $\alpha<1$ and we find that Thm.~\ref{thm:main} applies for computations of length $o(n^{\alpha \gamma})$.  Since we can choose $\gamma$ as desired, Thm.~\ref{thm:main} tells us that for any given polynomial scaling of the computation, there is a threshold below which fault-tolerant quantum computation with constant overhead is possible, but the value of the threshold depends on the scaling of the computation size with the number of logical qubits $k$ (although not on $k$ itself).  This is weaker than the result for codes with $n = \poly{d}$, for which we find that there is a fixed threshold that works for \emph{any} polynomial scaling of the size of the logical computation.

The advantage of the $2$-D hyperbolic surface codes is that there are standard algorithms for efficiently decoding surface codes, such as Edmonds' maximum matching algorithm~\cite{Edm}.  Edmonds' algorithm matches up pairs of defects so that the total distance between matched pairs is minimal.  This corresponds to picking the lowest-weight error (or rather a minimal-weight error, since it is frequently not unique) which could have caused the observed error syndrome.  The algorithm can also be applied to the case where syndrome measurements are not reliable by doing the matching on a graph triangulating a $3$-dimensional manifold, with the extra dimension representing time.  A defect is placed at (location $x$, time $t$) when the measured error syndrome bit associated to location $x$ changes at time $t$.  An optimal matching of defects assigned in this way is equivalent to a choice of a minimal-weight error history (time and location of errors, including ones on error syndrome measurements).  This is precisely the criterion used in Thm.~\ref{thm:adversarialFT}, which therefore shows that Edmonds' algorithm applied to a surface  code with $\log n$ distance suppresses errors polynomially, as discussed above, even when the error correction procedure is noisy.  Edmonds' algorithm runs in polynomial time, so Thm.~\ref{thm:main} applied to hyperbolic surface codes tells us that efficient fault tolerance can be achieved using polynomial classical computation, but unfortunately at the cost of having the threshold depend on the scaling of the logical computation time.

The $4$-D hyperbolic codes may present a satisfactory middle ground having the advantages of both the hypergraph product codes and of the $2$-D hyperbolic codes, but this is not yet certain.  The $4$-D hyperbolic codes have polynomial distance~\cite{GL}, so by Thm.~\ref{thm:adversarialFT} there exists a decoding algorithm which suppresses errors exponentially even with noisy error correction.  However, unlike the $2$-D case, we do not know an efficient implementation of this algorithm.  Recently, Hastings~\cite{Has13} has analyzed a different decoding algorithm which can be parallelized to $O(\log n)$ depth (total circuit size $O(n \log n)$) and has shown that it suppresses errors polynomially.  As with the $2$-D hyperbolic codes, this gives a threshold that depends on the scaling of the logical computation.  While Edmonds' algorithm runs in polynomial time for a $2$-D surface code, it is a somewhat large polynomial, so Hastings' algorithm for the $4$-D hyperbolic codes has a significant advantage in that respect.  Moreover, it is possible that Hastings' algorithm (or another decoding algorithm) for the $4$-D hyperbolic codes could actually suppress errors exponentially given the large distance of the codes.

The bottom line of this section is that there exist code families satisfying all the conditions of Thm.~\ref{thm:main}, but that none of the known code families is completely satisfactory.  Either we need to find better decoding algorithms for the existing code families with high distance, or we need to come up with new families of codes that have efficient decoding algorithms with exponential error suppression. 

\section{Error Correction}
\label{sec:FTEC}

One of the primary components of any fault-tolerant protocol is error correction.  There are three main methods of fault-tolerant error correction --- Shor error correction~\cite{ShorFT}, Steane error correction~\cite{SteaneEC}, and Knill error correction~\cite{KnillEC}.  Steane and Knill EC require large ancilla blocks, consisting of specific states encoded in the same error-correcting code being used for the main computation.  These large ancilla states must be built and tested, which typically requires additional overhead beyond the qubits in the ancilla itself, and error correction must be done on a constant fraction of the physical qubits used for coding data at any given time.  This suggests that Steane and Knill EC are not appropriate for the purpose of making a low-overhead fault-tolerant protocol.

That leaves Shor EC.  Shor EC, pictured in fig.~\ref{fig:ShorEC}, also uses ancilla states, but the ancilla state needed to measure the syndrome bit for a stabilizer generator $M$ is a \emph{cat state} $\ket{0}^{\otimes r} + \ket{1}^{\otimes r}$ containing a number of qubits $r$ equal to the weight of $M$.  To measure the full error syndrome of an $[[n,k,d]]$ QECC, we need $n-k$ such ancillas, and the syndrome measurement must be repeated many times so that we can rely on the results of the measurements.

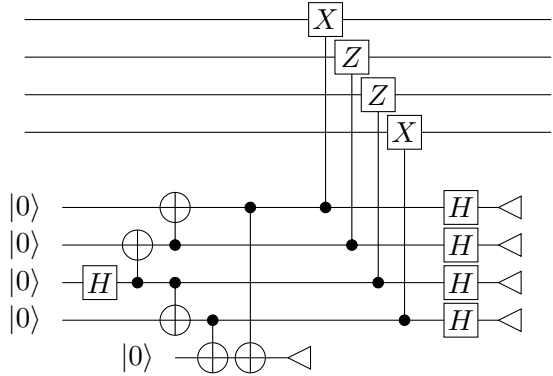
\begin{figure}
\begin{center}
\begin{tikzpicture}

\foreach \y in {4,4.5,...,5.5} {
	\draw (-0.5,\y) -- ++(7,0);}

\foreach \y in {1.5,2,2.5,3} {
	\draw (0,\y) -- (5.8,\y) -- ++(0.3,0.15) -- ++(0,-0.3) -- ++(-0.3,0.15);
	\node at (-0.5,\y) {$\ket{0}$};
	\draw [fill=white] (5.3,\y) ++(-0.225,-0.225) rectangle ++(0.45,0.45);
	\node at (5.3,\y) {$H$};
	}
\draw [fill=white] (0.275,1.775) rectangle ++(0.45,0.45);
\node at (0.5,2) {$H$};

\draw [fill=black] (1,2) circle (0.07cm) -- ++(0,0.7);
\draw (1,2.5) circle (0.2 cm);

\draw [fill=black] (1.5,2) circle (0.07cm) -- ++(0,-0.7);
\draw (1.5,1.5) circle (0.2 cm);

\draw [fill=black] (1.5,2.5) circle (0.07cm) -- ++(0,0.7);
\draw (1.5,3) circle (0.2 cm);

\node at (1,1) {$\ket{0}$};
\draw (1.5,1) -- ++(1.5,0) -- ++(0.3,0.15) -- ++(0,-0.3) -- ++(-0.3,0.15);

\draw [fill=black] (2,1.5) circle (0.07cm) -- ++(0,-0.7);
\draw (2,1) circle (0.2 cm);

\draw [fill=black] (2.5,3) circle (0.07cm) -- ++(0,-2.2);
\draw (2.5,1) circle (0.2 cm);

\draw [fill=black] (3.5,3) circle (0.07cm) -- ++(0,2.275);
\draw [fill=white] (3.275,5.275) rectangle ++(0.45,0.45);
\node at (3.5,5.5){$X$};

\draw [fill=black] (3.85,2.5) circle (0.07cm) -- ++(0,2.275);
\draw [fill=white] (3.625,4.775) rectangle ++(0.45,0.45);
\node at (3.85,5){$Z$};

\draw [fill=black] (4.2,2) circle (0.07cm) -- ++(0,2.275);
\draw [fill=white] (3.975,4.275) rectangle ++(0.45,0.45);
\node at (4.2,4.5){$Z$};

\draw [fill=black] (4.55,1.5) circle (0.07cm) -- ++(0,2.275);
\draw [fill=white] (4.325,3.775) rectangle ++(0.45,0.45);
\node at (4.55,4){$X$};

\end{tikzpicture}
\caption{The Shor fault-tolerant error-correction procedure applied to measure the syndrome bit for the generator $X \otimes Z \otimes Z \otimes X$.  First build and test the cat state, then interact it transversally with the codeword.}
\label{fig:ShorEC}
\end{center}
\end{figure}

The large ancillas used in Steane and Knill EC are codewords of a QECC, so once built, they can be sustained for a long time and a single-qubit error in one is not very serious.  Cat states, in contrast, are rather fragile.  A single phase error on any qubit in the cat state will result in the wrong outcome for a syndrome bit measurement in Shor EC, which is why we need to repeat the measurement many times.  For large cat states, the chance of a phase error is close to one.  Consequently,  Shor EC is impractical for codes with high-weight stabilizer generators.  If we are going to rely on Shor EC, we need to stick to LDPC codes.

Testing a single $r$-qubit cat state is done by checking the parity of pairs of qubits in the state.  The goal is to avoid correlated bit flip errors produced while the state is being built.  The total number of tests needed depends on $r$ and the precise circuit used to build the state, but the important thing is that it is independent of the size of the full code when $r$ is a constant.  Let $r'$ be the number of qubits in the cat state plus the number of test qubits needed.  A single measurement of the full syndrome of an $[[n,k,d]]$ $(r,c)$-LDPC code therefore uses a total of $(n-k) r'$ qubits.  Any set of generators which are non-overlapping can be done in parallel, so by appropriately partitioning the set of generators, a single measurement of the syndrome can be done in depth $O(r'c)$.  Repetitions of the syndrome measurement can reuse the same qubits, and therefore the total number of extra qubits involved in Shor FT EC is $(n-k)r'$.

It is worth noting that an alternative approach sometimes used for surface codes is to do error correction in a non-fault-tolerant way.  Instead of building a cat state and doing transversal gates to the code block, just use a single ancilla qubit for each syndrome bit and do the same set of controlled-Pauli gates all controlled by that ancilla qubit.  The drawback of this method is that a single fault in the procedure can propagate to multiple physical qubits of the code block.  However, for an $(r,c)$-LDPC code, the error can never propagate to more than $r$ physical qubits.  Applying the approach discussed in Sec.~\ref{sec:LDPC}, we see that a single fault can create only a small cluster of errors.  The results of that section therefore still apply but with a smaller threshold error rate.  On the other hand, this method involves fewer locations in a single syndrome bit measurement, so the effective physical error rate is lower too.  In some cases, it might be advantageous for the threshold to use this simpler procedure, and it is always advantageous for the overhead, reducing the number of extra qubits needed to just $n-k$ qubits per full syndrome measurement.  However, I will stick with Shor's method, as it is sufficient to prove Thm.~\ref{thm:main}.

The overhead can be reduced even further when the data is broken up into multiple code blocks.  Suppose we allocate the code blocks up into $s$ different sets, and we only perform a syndrome measurement on one set of blocks at a time, cycling through the sets systematically.  Each code block has to wait $s$ times as long between syndrome measurements, effectively increasing the storage error rate by a factor of $s$, but because we are only correcting a fraction of the qubits in the computation, the relative overhead needed decreases by a factor of $s$.  By choosing $s$ to be a sufficiently large constant independent of the total size of the computer, we can reduce the overhead from error correction as far as we like while still having a threshold.

\section{Gates, Preparation, and Measurement}
\label{sec:gates}

Performing logical gates on data encoded in a large block of a QECC can be difficult.  Certain types of gates are possible transversally or by permuting qubits, depending on the structure of the code.  For instance, for a CSS code, a transversal CNOT gate between two different code blocks will also perform a logical CNOT gate between corresponding logical qubits in the two blocks.  However, transversal gates can never be universal~\cite{EK}.  Specific codes might have additional tricks which let us perform certain non-transversal gates in a fairly straightforward way, but it would be difficult to rely on such an approach for any general result such as Thm.~\ref{thm:main}.

Luckily, techniques exist to perform a universal set of fault-tolerant gates for any stabilizer code~\cite{Got97}.  The most systematic approach is Knill's method~\cite{KnillEC}, which uses a particular logical state encoded in two blocks of the same stabilizer code used for the data.  The logical state is of the form $(I \otimes U) (\ket{00} + \ket{11})$, where $U$ is the gate we wish to perform, and the data is then teleported through this state as in gate teleportation~\cite{GC99}.  The data ends up in what was originally the second block of the ancilla.  A correction is then applied to the logical state depending on the result of the logical Bell measurement in the teleportation.  When $U$ is a Clifford group gate, the correction is just a logical Pauli operation (which can be done transversally for any stabilizer code).  For other gates $U$, a more complicated correction operation is needed, and we must be careful to pick gates for which the correction operation is always tractable.

When there are $k$ qubits encoded in each block, the same procedure works: Just use a logical ancilla state $(I \otimes U) \otimes (I \otimes I)^{k-1} (\ket{00} + \ket{11})^{\otimes k}$, applying $U$ to just the qubit on which we wish to perform the gate.  It can also work for two-qubit gates, using $4$ ancilla blocks if the two logical qubits we wish to interact are in separate code blocks (and similarly for three-qubit gates, etc.).  If the two logical qubits are in the same code block, simply use an ancilla of the same form, logical $(I \otimes U) \otimes (I \otimes I)^{k-2} (\ket{00} + \ket{11})^{\otimes k}$, but now with $U$ acting on the two logical qubits to be interacted.  For many specific gates, smaller ancillas may exist, but I will stick with this procedure because of its generality.

Knill's procedure allows one to make a logical Bell measurement between the first code block of the ancilla and the data block for any stabilizer code.  This can be done by just making a transversal Bell measurement between corresponding physical qubits of the first ancilla block and the former data block.  Then some classical processing is needed.  The measurement provides error correction information as well as the logical Bell measurement result.  Knill uses this technique to combine error correction with logical gates, but I am primarily interested in it for the purpose of performing gates.  However, the error correction information is still needed to ensure that the logical Bell measurement result is extracted reliably.

The drawback of this technique is that it needs large ancillas, as large as the code blocks we are using.  This means an $[[n,k,d]]$ code needs $2n$ or $4n$ extra ancilla qubits to do the gate, and reliably building such a large ancilla block is difficult and will likely require extra overhead.  When $U$ is a Clifford group gate, the ancilla state needed is a stabilizer state, but when $U$ is not a Clifford group gate (and it is essential to have at least one non-Clifford $U$ to have a universal set of gates), the ancilla state, known as a \emph{magic state}, is not a stabilizer state.  For specific stabilizer states and certain QECCs, there may be tricks that allow us to make large ancilla states without any extra overhead, but no code is known for which such tricks work for all the needed states.

One way to fault-tolerantly make any state is to rely on other fault-tolerant protocols.  Using concatenated codes, we can perform a universal set of gates provided the physical error rate per location is below a threshold value~\cite{AB97,Kit97,KLZ98}.  In particular, we can use a concatenated code to build any desired ancilla state $\ket{\Psi}$.  In our case, $\ket{\Psi}$ is $(I \otimes U) \otimes (I \otimes I)^{k-1} (\ket{00} + \ket{11})^{\otimes k}$ encoded in the main QECC, but the same procedure works for any state, as follows:  Take some non-fault-tolerant circuit $\circuit{D}$ that produces $\ket{\Psi}$.  Perform the fault-tolerant simulation of $\circuit{D}$ using concatenated codes.  The result (with arbitrarily high probability) is a copy of $\ket{\Psi}$ with each physical qubit of $\ket{\Psi}$ encoded using a separate block of a concatenated QECC.  Each concatenated code block can then be decoded level-by-level, resulting in an unencoded copy of $\ket{\Psi}$, as shown in fig.~\ref{fig:FTprep}.  The procedure results in a physical error rate per physical qubit of $\ket{\Psi}$ which is bounded by some constant value~\cite{KL97}.  A similar procedure was used for some fault-tolerant gates in \cite{AB97,AB08}.

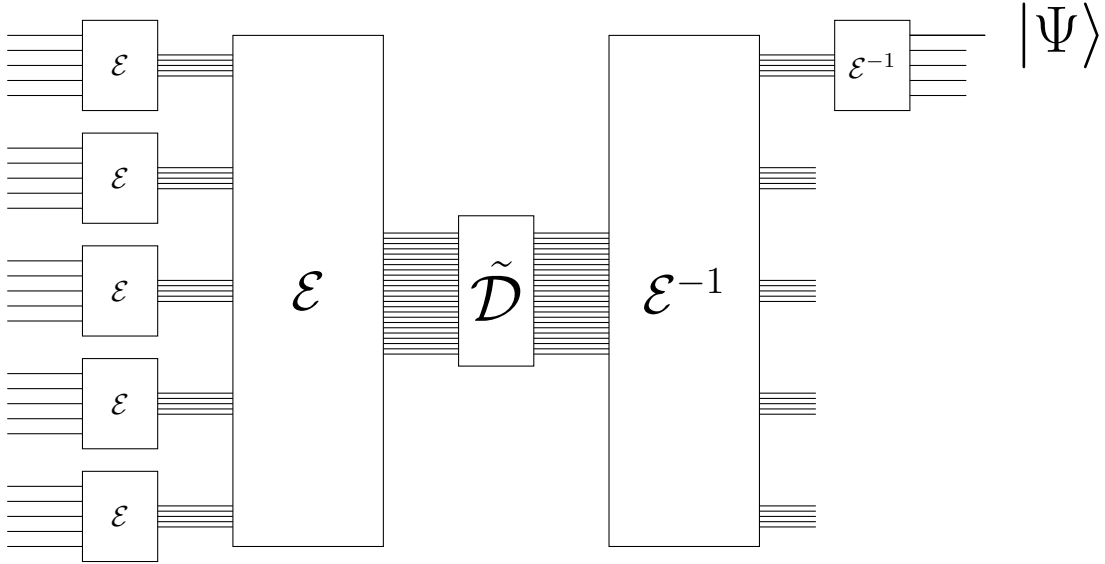
\begin{figure}
\begin{center}
\begin{tikzpicture}

\foreach \bigy in {3,4.5,...,9} {
	\foreach \smally in {0,0.2,...,0.8} {
		\draw (0,\bigy) ++(0,-0.4) ++(0,\smally) -- ++(1,0);
	}
	\draw (1,\bigy) ++(0,-0.6) rectangle ++(1,1.2);
	\node at (1.5,\bigy) {$\circuit{E}$};
	\foreach \smally in {0.07,0.14,...,0.35} {
		\draw (2,\bigy) ++(0,-0.21) ++(0,\smally) -- ++(1,0);
	}
}

\draw (3,2.6) rectangle ++(2,6.8);
\node at (4,6) {{\huge $\circuit{E}$}};

\foreach \smally in {0,0.07,...,1.68} {
	\draw (5,6) ++(0,-0.84) ++(0,\smally) -- ++(1,0);}

\draw (6,5) rectangle ++(1,2);
\node at (6.5,6) {{\Huge $\circuit{\tilde{D}}$}};

\foreach \smally in {0,0.07,...,1.68} {
	\draw (7,6) ++(0,-0.84) ++(0,\smally) -- ++(1,0);}

\draw (8,2.6) rectangle ++(2,6.8);
\node at (9,6) {{\huge $\circuit{E}^{-1}$}};

\foreach \bigy in {3,4.5,...,7.5} {
	\foreach \smally in {0.07,0.14,...,0.35} {
		\draw (10,\bigy) ++(0,-0.21) ++(0,\smally) -- ++(0.75,0);
	}
}
\foreach \smally in {0.07,0.14,...,0.35} {
	\draw (10,9) ++(0,-0.21) ++(0,\smally) -- ++(1,0);
}
	\draw (11,9) ++(0,-0.6) rectangle ++(1,1.2);
	\node at (11.5,9) {$\circuit{E}^{-1}$};
	\foreach \smally in {0,0.2,...,0.6} {
		\draw (12,9) ++(0,-0.4) ++(0,\smally) -- ++(0.75,0);
	\draw (12,9) ++(0,0.4) -- ++(1,0);
	}

\node at (14,9.4) {{\Huge $\ket{\Psi}$}};

%\draw [decorate,decoration={brace,amplitude=20pt}] (13.2,9.6) -- ++(0,-7.2);

\end{tikzpicture}
\caption{A technique to fault-tolerantly create any ancilla state $\ket{\Psi}$: Encode a concatenated code; each level uses the encoding circuit $\circuit{E}$.  Using a fault-tolerant protocol for the concatenated code, run a fault-tolerant version $\circuit{\tilde{D}}$ of a circuit $\circuit{D}$ which creates $\ket{\Psi}$.  Then decode the concatenated code.}
\label{fig:FTprep}
\end{center}
\end{figure}

The catch is that this procedure uses non-negligible overhead.  We have to assume that a single logical error during the circuit $\circuit{D}$ can cause the state $\ket{\Psi}$ to be arbitrarily wrong.  Therefore, if we want the probability of a logical error for the whole encoding circuit to be $\epsilon_0$, we need to use enough levels of concatenation so that the logical error rate per concatenated gate is $\epsilon_0/|\circuit{D}|$, with $|\circuit{D}|$ the number of locations in $\circuit{D}$.  Concatenated coding requires polylog overhead, so to use this method to make $\ket{\Psi}$, we need a total number of qubits equal to $O(n \polylog(|\circuit{D}|/\epsilon_0))$.  All of the ancilla states $\ket{\Psi}$ that we need have encoding circuits that are polynomial (quadratic, in fact) in their size, so the total number of qubits we need is $O(n \polylog (n/\epsilon_0))$.

If we encode all of the logical qubits into a single $[[n,k]]$ code block, the overhead is too much --- more than a constant factor.  However, if we split the logical qubits up into multiple code blocks, each of parameters $[[n',k']]$, things work out much better.  If we do only a single logical gate at a time, we only need a single ancilla state, and the number of extra qubits needed to create the ancilla and perform the gate is instead $O((n' \polylog(n'/\epsilon_0))$.  When $n'$ is less than $n/\polylog (n/\epsilon_0)$, the number of extra qubits needed for the gate is sublinear in $n$ and does not contribute to the overall overhead of the computation.  Note, though, that we can only perform one such gate (or a small number) at a time.  This is why Knill's method is still useful for performing gates even though we could not use it for error correction, which must be done on a large fraction of the computer at once.

Also relevant is the depth of the circuits needed to create the ancilla state and perform the logical gate.  Even if we only do one logical gate at a time, we may want to prepare the ancilla states for upcoming gates in parallel.  Typically a concatenated code blows up the depth of a circuit $\circuit{D}$ by a factor of $\polylog(|\circuit{D}|/\epsilon_0)$ (though not necessarily the same polynomial).  In this model, with free classical computation, the depth blow-up can actually be reduced to a constant factor plus an additive $O(\log \log (|\circuit{D}|/\epsilon_0))$ term (see Sec.~\ref{sec:depth} for further details), but that is not important at the moment.  The depth of the circuit to create an arbitrary state $\ket{\Psi}$  might be large, but the ancilla states we need simply consist of a tensor product of one- or two-qubit states which are then encoded using a stabilizer code.  Every stabilizer code can be encoded using a parallel circuit~\cite{MN98}, so the overall depth of the ancilla preparation procedure is just $O(\polylog(|\circuit{D}|/\epsilon_0))$.  

Normally, a parallel version of $\circuit{D}$ would require $\poly(|\circuit{D}|)$  extra qubits, which would put a more severe constraint on $n'$ to keep the total number of ancilla qubits sublinear in $n$.  Since we are dealing with an LDPC code, however, it is possible to perform the encoding circuit $\circuit{D}$ using constant quantum depth, $O(\log n')$ classical depth, and less than $2n'$ total qubits (each of which will be encoded in the concatenated code).  First, I claim that for any $[[n',k']]$ stabilizer code, up to permutation of the qubits and single-qubit Hadamards, it is possible to choose the logical operators to be $\overline{X}_i = X_i \otimes P_i$ and $\overline{Z}_i = Z_i \otimes Q_i$, where $X_i$ and $Z_i$ are $X$ and $Z$ acting on the $i$th qubit and $P_i$ and $Q_i$ are tensor products of $Z$ and $I$ acting on the last $n'-k'$ qubits of the code.  To see this, note that by Gaussian elimination and Hadamard transforms it is possible to bring the stabilizer into the form $(AI|BC)$, represented as a binary symplectic matrix.  We can then choose $P_i$ and $Q_i$ as a tensor product of $Z$'s acting on the last $n'-k'$ qubits to commute with the same subset of generators as $X_i$ and $Z_i$ do ($i = 1, \ldots, k'$), so that $\overline{X}_i$ and $\overline{Z}_i$ commute with everything in the stabilizer.  They also have the correct commutation relations with each other, proving the claim.

We can then encode a state in a non-fault-tolerant way by placing the qubits to be encoded as the first $k'$ qubits of the code, with $\ket{0}$ for the last $n'-k'$ qubits, and then measuring the stabilizer generators.  The resulting state is then the logical input state encoded in some coset of the code, which can be shifted back to the correct code with an appropriate Pauli operator.  All of this was done for a code with some qubits Hadamard transformed, so we could do this encoding circuit and then get the original code with some single-qubit Hadamards.  Note that we don't need to measure the stabilizer generators in the canonical form $(AI|BC)$; in particular, we can measure the LDPC presentation of the generators.  As in Sec.~\ref{sec:FTEC}, a single syndrome measurement can be done in constant depth.  Because all of this is done using qubits encoded with the concatenated code, we have no need to make a fault-tolerant measurement, so we only need $n'-k'$ ancilla qubits and don't need to repeat the syndrome measurement.  The quantum depth is thus constant; to determine the appropriate Pauli operator is  a linear classical operation, which can be done in depth $O(\log n')$.  Therefore, even if we wish to prepare upcoming ancilla states in parallel, the total number of qubits needed at any given time is still $O((n' \polylog(n'/\epsilon_0))$, just with some slightly larger constants.

The same procedure used to generate ancilla states can also be used to generate logical $\ket{0}$ states at the start of the computation.  At the start of the computation, generate one block's worth of new logical qubits at a time.  Logical qubits that are created early must wait until all needed logical qubits are ready.  Of course, they are subject to storage errors, so we must perform error correction on them while they are waiting.

We will assume that the logical circuit has been modified so that we measure all the logical qubits in a block at the same time.  This can be done either by delaying all logical measurements until the end of the circuit or by copying qubits that need to be measured into an auxiliary block with a CNOT.  Measurement in the standard computational basis is straightforward for a CSS code: Simply measure each physical qubit, getting a noisy classical codeword, and then decode the classical code.  This can be done in one time step (plus classical computation) with no extra qubits.  For a general stabilizer code, it is a bit more complicated.  A version of Knill's procedure can be used, or one can simply use a version of the state preparation procedure in reverse: Encode each physical qubit of the code block using a concatenated code.  Design a circuit $\circuit{D'}$ that performs error correction for the main QECC and decodes the code block, and run the fault-tolerant simulation of $\circuit{D'}$ for concatenated coding on the state.  Then use the fault-tolerant measurement procedure for the concatenated code to extract the answer.  Either way, this uses $O(n' \polylog (n'/\epsilon_0))$ extra qubits to do measurements on a general stabilizer code, so we should only do measurement of a single block at a time.

\section{Combining the Components}
\label{sec:combine}

Now I will combine the observations from Secs.~\ref{sec:FTEC} and \ref{sec:gates} in order to prove the main theorem.

\begin{proof}[Proof of Thm.~\ref{thm:main}]
As I discussed in Sec.~\ref{sec:gates}, we don't wish to encode all logical qubits in the same block of the QECC.  Given the family of codes $Q_i$ which are $[[n_i,k_i]]$ codes and a circuit $\circuit{C}$ which we wish to perform, we choose an $i$ so that the code $Q_i$ is sufficiently small that polylog overhead on a single block is not a problem but is sufficiently large that errors are still rare.  I will discuss precisely how to choose $i$ shortly.  There will be a total of $M = \lceil k/k_i \rceil$ blocks of the QECC.  

By condition~(\ref{item:adversarial}) in the statement of the theorem, the probability of a logical error during syndrome decoding, provided we stay below the listed error thresholds, is at most $D/g(n_i)$ for some constant $D$.  To apply condition~(\ref{item:adversarial}), we require that all errors be due to local stochastic noise, that the input to each error correction cycle have an error rate $\tilde{p}$ less than $p_0$, that each physical qubit in a data block suffers an error rate $p_D$ of at most $p_1$ per syndrome measurement during error correction, and that each measured syndrome bit has an error rate $q$ at most $p_2$.    For each syndrome measurement, we use these error rates to produce new data errors and errors in syndrome bits.

However, this applies in a simplified model of fault tolerance, and now we are working with the basic model of fault tolerance, in which the error correction circuit is broken down into locations representing single gates (or other actions) and each location can have a fault.  To make the connection between the two models, for every fault that can occur in the basic model, we identify which data qubit and which syndrome bits it causes errors in.  In particular, we associate each basic model fault with a particular syndrome measurement, and in the simplified model, that corresponds to either a syndrome bit error during that measurement or a data qubit error just before that syndrome measurement, plus some syndrome bit errors during the measurement.  The simplest way of doing so is to associate a fault with syndrome measurement $i$ if it occurs during the circuit used to measure the $i$th repetition of the syndrome.  If a fault occurs during the measurement of a particular bit of the error syndrome, all bits measured afterwards are consistent with the new error syndrome rather than the old one.  The syndrome bits that were measured before the fault do not take into account the new error.  In this case, we can consider the fault as causing an error on both the physical data qubit affected by it and on the syndrome bits that did not take account of the new error. A single fault in a gate could therefore cause errors on a single physical data qubit plus up to $c$ syndrome bits (the maximum number of stabilizer generators involving that qubit).  All the erroneous syndrome bits will be on the syndrome measurement associated with the fault.

A more efficient approach is to instead associate a fault with the subsequent syndrome measurement (number $i+1$) if it affects a data qubit and occurs after half the syndrome bits involving the qubit have been measured.  Then the corresponding data qubit error occurs just after the $i$th syndrome measurement, but the syndrome bit errors resulting from the fault are still in the $i$th syndrome measurement (before the fault occurred).  Now the bits measured \emph{after} the fault are incorrect, since they correspond to the new error even though it hasn't happened yet (in the simplified model).  There are at most $c/2$ such erroneous syndrome bits.  If the fault occurred while measuring a particular syndrome bit, that bit could be wrong as well due to the fault and be consistent with neither the old nor the new syndrome.  We can absorb this into the $c/2$ possible wrong syndrome bits unless the bit in question is exactly in the middle.  Thus, a single fault can cause up to $\lceil c/2 \rceil$ syndrome bit errors.  If a fault is shifted beyond the end of an error-correction cycle, we include it instead as an initialization error in the next error correction cycle.

There are three possible causes of an error on a single physical qubit in the QECC block.  The first possibility is a single-qubit error which propagates into the data block from one of the ancillas used to measure syndrome bits including that qubit.  Such an error was originally created during the preparation of the ancilla state.  Suppose there are $A$ possible locations for an error in each ancilla preparation circuit.  ($A$ is a function of $r$, which determines the size of the cat state.)  There are at most $c$ ancillas interacting with a single data qubit since each qubit is contained in at most $c$ stabilizer generators, for a total of $cA$ possible locations for cat state preparation errors.  The second possibility is a fault in one of the gates used to measure a syndrome bit on the qubit.  There will be at most $c$ such gates.  The third possibility is a storage error while the qubit is waiting to be measured.  Suppose there are $l$ time steps during a single syndrome measurement.  Each code block is measured only once out of every $s$ times, so the total waiting time could be up to $sl$ time steps.  One small technical correction is needed to the probability of an ancilla error if we post-select ancilla states that pass a test.  The correction is to divide by the probability of successful post-selection, which is at least $1-pA$ and is extremely close to $1$ for typical fault-tolerant error rates.  I will assume that errors are not explicitly corrected, but that instead we keep track of the current Pauli frame of the computation (the total Pauli operator that has been applied), and update that based on errors that are identified.  The total probability of an error on a physical data qubit is therefore
\begin{equation}
p_P = (cA/(1-pA) + c + sl)p.
\end{equation}
Correlated errors are possible even for independent noise because syndrome measurements indirectly interact different qubits in the block.  If we use Shor error correction, the procedure is fault-tolerant, so the qubits in the data block experience local stochastic noise --- the probability that a particular set of $a$ different data qubits all have errors is at most $p_P^a$.  However, as noted above, an error on a data qubit in the middle of a syndrome measurement could cause up to $\lceil c/2 \rceil$ syndrome bits to be wrong as well.

In addition to errors due to the syndrome changing during a measurement, an error on a syndrome bit can be caused directly by an error in the preparation of the ancilla used to measure that syndrome bit, by a fault in one of the controlled-Pauli gates used to do the measurement, or a fault in one of the Hadamard rotations or measurements after the controlled-Pauli gates.  With Shor error correction, the ancilla is a cat of at most $r$ qubits (the maximum size of a stabilizer generator).  The total probability of a direct syndrome bit error is
\begin{equation}
p_B = (A/(1-pA) + 3r)p.
\end{equation}
The direct syndrome bit errors are uncorrelated with each other.

Because of the correlations between syndrome bit errors and data qubit errors, we must be careful how we choose the effective error rates $p_D$ and $q$ to use in condition~(\ref{item:adversarial}) so that we have a local stochastic noise model.  We can label faults as causing physical qubit errors, which can also result in multiple syndrome bit errors as discussed above, or as direct syndrome errors, those without a corresponding error on a data qubit.  The probability of having a particular set of $a_P$ physical qubit errors and $a_B$ direct syndrome bit errors is at most $p_P^{a_P} p_B^{a_B}$.  Such an error will show up (in the counting for condition~(\ref{item:adversarial})) as $a_P$ data qubit errors and up to $a_P \lceil c/2 \rceil + a_B$ syndrome bit errors.  If we pick a particular set $S$ of $b_D$ data qubit errors and $b_B$ syndrome bit errors in the simplified fault-tolerance model, we should sum over $a_B$, since a physical data qubit error causes \emph{at most} $\lceil c/2 \rceil$ syndrome bit errors, but could cause less.  The minimum value of $a_B$ is $a_{\mathrm{min}} = b_B - b_D \lceil c/2 \rceil$, and the maximum value is $b_B$.  Therefore, the total probability of errors on the set $S$ is at most
\begin{equation}
\sum_{a_B = a_{\mathrm{min}}}^{b_B} \binom{b_B}{a_B} p_P^{b_D} p_B^{a_B} \leq \sum_{a_B = a_{\mathrm{min}}}^{b_B} 2^{b_B} p_P^{b_D} p_B^{a_B} \leq (2^{\lceil c/2 \rceil} p_P)^{b_D} \frac{(2 p_B)^{a_{\mathrm{min}}}}{1- p_B}.
\end{equation}
This bound holds even if $a_{\mathrm{min}}<0$.

We must choose error rates for the simplified fault-tolerance model.  I will use $p_D$ for the data qubit error rate in the simplified model and $q$ for the syndrome bit error rate.  We wish the probability of errors on the set $S$ to be at most 
\begin{equation}
p_D^{b_D} q^{b_B} = (p_D q^{\lceil c/2 \rceil})^{b_D} q^{a_{\mathrm{min}}}.
\end{equation}
It thus is sufficient if we choose $p_D$ and $q$ so that
\begin{align}
p_D q^{\lceil c/2 \rceil} &\geq \frac{2^{\lceil c/2 \rceil} p_P}{1-p_B}, \\
q &\geq \frac{2 p_B}{1-p_B}.
\end{align}
For instance, it suffices to let 
\begin{align}
p_D &= \frac{\sqrt{p_P}}{1-p_B}, \\
q &= \max \left( 2 p_P^{1/(c+1)}, \frac{2p_B}{1-p_B} \right).
\end{align}
The first two threshold conditions then become:
\begin{align}
p_D &< p_1, \label{eqn:p1threshold} \\
q &< p_2. \label{eqn:p2threshold}
\end{align}

Errors on the qubits entering an error correction cycle could have survived from the end of the previous error correction cycle or they could be due to errors in the logical gate immediately preceding the error correction step.  Errors from the previous cycle will affect a qubit if they are in the same qubit from the corresponding block in the previous cycle or if they propagate into the current block from the other block involved in the gate (in the case of a gate interacting two different code blocks).  Knill's teleportation method allows a simultaneous gate and error correction, but I will neglect the effects of this error correction to simplify the analysis.  The ancilla is prepared via a fault-tolerant circuit for a concatenated code, as described in Sec.~\ref{sec:gates}.  This leads to a small probability of complete failure of the ancilla, plus local stochastic noise with an error rate $Bp$ per physical qubit, $B$ a constant depending on the details of concatenated code decoding circuit.  If the gate is a Clifford group gate, the teleportation procedure requires a physical Bell measurement, which I will break down to $3$ physical locations (Hadamard, CNOT, measurement) per qubit, plus a logical Pauli operator, a total of $4$ physical locations per qubit.  These are applied transversally, so the probability of having faults of this type on multiple qubits is given by the product; i.e., we still have local stochastic noise.  When the gate is a non-Clifford group gate, for instance the $\pi/8$ phase rotation, another Clifford group gate is needed to correct for the teleportation outcome.

Let us assume that the previous error correction cycle succeeds, which happens with probability at least $1-D/g(n_i)$ per block.  Then by condition~(\ref{item:adversarial}), the output qubits from the previous cycle have errors governed by a local stochastic noise model with error rate at most $p_0/3$.  For a CNOT gate, there are two prior blocks, plus the additional error rate $(B+4)p$ per qubit, so we require that
\begin{equation}
2p_0/3 + (B+4)p < p_0.
\label{eqn:p0thresholdCNOT}
\end{equation}
The $\pi/8$ gate is a single-qubit gate, so there can only be one block involved in the computation.  Two teleportations might be needed, however, so the condition for a $\pi/8$ gate is
\begin{equation}
p_0/3 + 2(B+4)p < p_0.
\label{eqn:p0thresholdmagic}
\end{equation}
The logical single-qubit Clifford group gates and the logical wait location have lower error rates and therefore less stringent conditions.

Equations~(\ref{eqn:p1threshold}), (\ref{eqn:p2threshold}), (\ref{eqn:p0thresholdCNOT}), and (\ref{eqn:p0thresholdmagic}) define the threshold error rate $p_T$ for the overall protocol.  We choose $p_T$ so that whenever $p<p_T$, all these inequalities are satisfied.  The argument I have described above uses a somewhat different approach from the threshold proof in \cite{AGP06}; in terms of the tools from that paper, one could phrase the argument here as a process that replaces 
\begin{equation}
\text{ideal encoder --- initiatialization errors --- gate --- EC} 
\end{equation}
with
\begin{equation}
\text{ideal gate --- ideal encoder --- initialization errors.}
\end{equation}
The initialization errors can come from the decoding of a concatenated code used for state preparation or from pushing the ideal encoder forward through the previous EC cycle.

The circuit $\circuit{C}$ uses $k$ logical qubits and $f(k)$ locations, and we want the overall probability of error to be at most $\epsilon$.  We therefore wish the logical error rate per logical location to be at most $\epsilon/f(k)$.  If the location involves state preparation (either to do a gate or because that is the primary purpose of the location), then there is a probability $\epsilon_0$ that the preparation itself fails, and similarly if we need to do a measurement for a general stabilizer code.  If the gate to be performed is a $\pi/8$ phase rotation, we may need two ancilla states.  Let us therefore set $\epsilon_0 = \epsilon/(3f(k))$.  There is also a probability $D/g(n_i)$ that the logical location fails because of too many errors for the error correction procedure to handle, and we choose $n_i$ so that 
\begin{equation}
D/g(n_i) \leq \epsilon/(3f(k)).  
\label{eqn:zerothk0}
\end{equation}
In particular, since $f(k) = o(g(k^{\alpha}))$ for some $\alpha < 1$, it suffices to choose $n_i > k^{\alpha}$ for sufficiently large $k$.  With this choice, the total logical error rate per logical location is bounded, as desired, by $\epsilon/f(k)$.

The total number of extra qubits involved in the circuit, in addition to the $Mn_i$ qubits used to encode the data, is the sum of the total number needed for error correction at any given time, plus the total number devoted to performing locations.  From Sec.~\ref{sec:FTEC}, the total number of qubits needed for error correction is $M(n_i-k_i)r'/s(1-Ap)$, where we have chosen to perform error correction on only one out of every $s$ code blocks at a time in order to save on overhead.  (We need to prepare an average of $1/(1-Ap)$ cat states for each one that we use to take account of failed attempts to create a cat state.)  Suppose that our methods for fault-tolerant state preparation and measurement (including whatever amount of parallel preparation and measurement we wish to do) use at most $C n_i (\log (n_i/\epsilon_0))^a$ qubits at any given time.  
Let us choose 
\begin{equation}
n_i < (k/R) (\log (k/R\epsilon_0))^{-(a+1)} = (k/R) (\log (3kf(k)/R\epsilon))^{-(a+1)}.  
\label{eqn:nirange}
\end{equation}
Then the number of extra qubits needed for gate, preparation, or measurement locations is at most
\begin{equation}
\frac{C k}{R} \left(\log \frac{k}{R\epsilon_0}\right)^{-(a+1)} \left(\log \left[\frac{k\left(\log \frac{k}{R\epsilon_0}\right)^{-(a+1)}}{R\epsilon_0}\right]\right)^a = \frac{C k}{R \log (k/R\epsilon_0)} \left[1 + O\left(\frac{\log \log (k/R\epsilon_0)}{\log (k/R\epsilon_0)}\right)\right].
\end{equation}
For sufficiently large $k$, 
\begin{equation}
\frac{C k}{R \log (k/R\epsilon_0)} \left[1 + O\left(\frac{\log \log (k/R\epsilon_0)}{\log (k/R\epsilon_0)}\right)\right] < \frac{Ck}{Rs'}
\label{eqn:firstk0}
\end{equation}
for any desired $s'$.  We will also need for $Mk_i \approx k$ so that not too many qubits are wasted when we divide up the logical qubits into blocks.  In particular, we wish $Mk_i \leq k(1 +\lambda)$, where $\lambda > 0$ and $1+\lambda < \eta$. (Recall that $\eta/R$ is the desired overhead of the computation.)  Let 
\begin{equation}
R/(1+\omega) \leq k_i/n_i \leq R(1+\omega).
\label{eqn:rateR}  
\end{equation}
Since $k_i/n_i \rightarrow R$, as $k$ gets large (so that $i$ gets large), $\omega > 0$ can be made as small as desired.

The total number of qubits needed for the computation is thus at most
\begin{equation}
\frac{Mk_i(1+\omega)}{R} + \frac{Mk_i ((1+\omega)-R)r'}{R(1-Ap)s} + \frac{Ck}{Rs'} \leq \frac{k}{R} \left((1 +\lambda)(1+\omega) + \frac{(1+\lambda)(1+\omega-R)r'}{(1-Ap)s} + \frac{C}{s'} \right) < \frac{\eta k}{R}
\label{eqn:totalqubits}
\end{equation}
for suitable choice of $s$ and $s'$ and small enough $\omega$.

In order for this to work, we need to choose an $i$ such that
\begin{equation}
k^{\alpha} < n_i < (k/R) (\log (3kf(k)/R\epsilon))^{-(a+1)}
\label{eqn:choosei}
\end{equation}
and $1 \leq M (k_i/k) \leq 1 + \lambda$.  $f(k)$ is a polynomial, so the condition says that $n_i = O(k/\polylog(k))$.  Now, suppose we satisfy equation~(\ref{eqn:choosei}).  $M \leq k/k_i + 1$, so $Mk_i \leq k + k_i$, so we need $k_i/k \leq \lambda$.  We have 
\begin{equation}
k_i \leq n_i R (1 + \omega) < k (1+ \omega) (\log (3kf(k)/R\epsilon))^{-(a+1)} < k (1 + \lambda)
\label{eqn:secondk0}
\end{equation}
for sufficiently large $k$.  Therefore, (\ref{eqn:choosei}) is sufficient for large enough circuits.

Choose $i$ to be the smallest integer such that $n_i > k^{\alpha}$.  Then by condition~(\ref{item:slowgrowth}) of the theorem,
\begin{equation}
n_i < n_{i-1} + n_{i-1}^\beta \leq k^\alpha + k^{\alpha \beta} = o(k/\polylog(k)),
\label{eqn:thirdk0}
\end{equation}
since $\alpha, \alpha \beta < 1$. 
Thus, for sufficiently large $k$, this choice of $i$ will cause $n_i$ to satisfy equation~(\ref{eqn:choosei}).  The threshold computation size $k_0$ is given by equations~(\ref{eqn:zerothk0}), (\ref{eqn:firstk0}), (\ref{eqn:rateR}), (\ref{eqn:secondk0}), and (\ref{eqn:thirdk0}); when $k>k_0$, these are satisfied.  Then $n_i$ lies in the desired range (\ref{eqn:choosei}), the qubit overhead of the computation is at most $\eta/R$, and the overall logical error rate of the computation is below $\epsilon$. Let $\alpha' = \max(\alpha, \alpha \beta)$.  Then $n_i < 2k^{\alpha'}$.

The classical computation used in the protocol is that needed to decode syndromes for the code $Q_i$, for the concatenated code used to perform state preparation (and possibly logical measurement), and for the teleportations used for logical gates.  The syndrome decoding of $Q_i$ uses a depth $h(n_i) < h(2k^{\alpha'})$ classical circuit for each error correction cycle.  The concatenated code has $O(\log \log \epsilon_0) = O(\log \log (f(k)/\epsilon))$ levels, and error correction for it can be performed in a constant depth per level.  The teleportation uses primarily the syndrome decoding for the QECC, plus a small constant number of time steps.

A single logical location in the fault-tolerant circuit requires one or two error correction cycles plus up to two ancilla preparations and up to $8$ transversal gates.  A transversal gate has $n_i$ locations.  A single error correction cycle uses $T(n_i)$ syndrome measurements, each of which involves $sl$ time steps and $n_i - k_i$ measurements using a cat state which took $A/(1-Ap)$ gates to create on average (taking into account cat states that fail the test) and $2r$ gates to measure.  (The controlled-Pauli gates interacting the cat state with the data block are included in the $sl$ time steps for the main block.)  This is a total of
\begin{equation}
T(n_i) [sl n_i + (n_i - k_i) (A/(1-Ap)+2r)]
\end{equation}
locations per error correction cycle.  One ancilla preparation for a gate uses $O(n_i \polylog (n_i/\epsilon_0))$ locations.  The total number of physical locations per logical location is thus 
\begin{equation}
O(sT(n_i) n_i + n_i \polylog (n_i f(k)/\epsilon)).  
\end{equation}
Again, $n_i < 2k^{\alpha'}$.  $s$ must be chosen so that
\begin{equation}
(1+\lambda)(1/R - 1) r' < s (\eta-1).
\end{equation}
(In fact, we have the more stringent condition given by (\ref{eqn:totalqubits}), but this is sufficient to understand the scaling of $s$.)  $r'$ is some constant dependent only on $r$ and $\lambda$ we choose to be fairly small.  $R$ we similarly neglect as constant, but $\eta$ we want to consider as variable.  Consequently, $s = \Theta(1/(\eta-1))$.  The total number of physical locations in the circuit is thus
\begin{equation}
O \left( \frac{f(k) T(2k^{\alpha'}) k^{\alpha'}}{\eta - 1} + f(k) k^{\alpha'} \polylog (k f(k)/\epsilon) \right).
\end{equation}
\end{proof}

\section{Depth of the Fault-Tolerant Circuits}
\label{sec:depth}

I have now shown that under the right conditions, it is possible to do fault-tolerant quantum computation with constant space overhead, and potentially a very low constant at that.  However, this protocol does not also achieve low overhead in \emph{time}.  There are two main sources of a blow-up in the depth of the fault-tolerant circuit relative to the ideal circuit.  First, we must repeat each syndrome measurement $T(n_i)$ times and using Thm.~\ref{thm:adversarialFT} requires that $T(n_i) = d_i$, which grows as a fractional power of $n_i$.  Second, Thm.~\ref{thm:main} is based on the assumption that we are implementing a \emph{sequential} circuit, with only one gate happening at a time.  If we are given a parallel quantum algorithm instead, the depth of the fault-tolerant simulation that comes out of Thm.~\ref{thm:main} is based on the \emph{size} of the original circuit (the number of gates) rather than the depth of the original circuit.

It is probably possible to remove most of the blow-up in depth due to repeating the syndrome measurement.  The protocol as described in this paper assumes we measure the syndrome $T(n_i)$ times and determine the errors, then do a logical gate, then measure the syndrome $T(n_i)$ times again, and so on.  After each logical gate, we do a full cycle of syndrome measurements, interpret the results, and then discard them and start over after the next logical gate.  However, this is not the most efficient thing to do.

For any block that does not experience a logical gate, we can retain syndrome information from the previous error correction cycle.  This results in better performance of the decoding algorithm, as noted in Sec.~\ref{sec:LDPC}, and, in particular, allows us to get away with just a single syndrome measurement in each error correction cycle.  By looking back over the last $T(n_i)$ error correction steps for this block, we accumulate enough syndrome information to reliably correct the block.  

A block on which a gate is performed might not work this way, but recall that Knill's method of performing gates allows error correction to be performed simultaneously.  Knill's method is not as susceptible to syndrome bit errors as Shor EC (an error in a single qubit during Knill EC instead just masquerades as an error in that particular qubit), so the syndrome measurement does not need to be repeated like it does in Shor EC.  In short, it should be sufficient to perform Knill EC on the code blocks with gates and Shor EC on the code blocks without gates, collecting syndrome information over many time steps.  Therefore the error correction only produces a constant blow-up, by a factor of about $sl$, the time to measure the syndrome once on all blocks.  There is an explicit trade-off here between space and time --- larger $s$ means fewer qubits are needed for error correction at any given time, but it also means we must spend more time waiting.

Removing the need for the original circuit $C$ to be sequential is much harder, but at least the condition can be relaxed substantially.  Suppose we have a circuit that performs at most $k^\gamma$ gates simultaneously.  This requires $O(k^\gamma n_i \polylog(n_i/\epsilon_0))$ extra qubits. Recall that $n_i < 2k^{\alpha'}$, so provided $\gamma + \alpha' < 1$, the circuit still uses a sub-linear number of qubits for gates at any given time, and the overall asymptotic effect on the overhead is negligible.  For some families of codes that we might consider, $g(n_i)$ will be super-polynomial.  For instance, in Thm.~\ref{thm:adversarialFT}, $g(n_i) = \exp(d_i)$, with $d_i$ polynomial in $n_i$.  Therefore, any $\alpha > 0$ will suffice for the main theorem to cover any polynomial-length computation.  Even for more marginal code families, such as the hyperbolic surface codes, we can take any $\alpha > 0$ and any polynomial scaling $f(n)$ for the computation size by lowering the threshold.  The condition on $\beta$ is also fairly loose.  For instance, if $n_i = 2^i$, the family of codes $Q_i$ is fairly sparse, but this still allows any $\beta > 1$.  Consequently, it is quite reasonable to set $\alpha$ and $\alpha \beta$ arbitrarily close to $0$, and $\gamma$ can be arbitrarily close to $1$, allowing for quite widespread parallelism in $C$ and a time overhead of only $k^{1-\gamma}$.  There is a disadvantage to doing so, because the point of setting $\alpha'$ to be small is to choose small block sizes $n_i$.  This means that there is not too much error suppression -- the probability of a block failing is $1/g(n_i)$.  For large enough blocks, it is good enough error suppression, but small $n_i$ means we need larger computations (larger $k_0$) before the fault tolerant protocol shows its advantage.

Furthermore, complete parallelism ($\gamma = 1$) seems impossible using the approach I have discussed in this paper.  Because there is polylog overhead relative to the block size for the state preparation procedure, simultaneously making ancillas for all code blocks with this method necessarily requires polylog overhead for the full computation.  To get around this would require a family of codes with additional properties allowing either direct implementation of a universal set of fault-tolerant gates with constant overhead or a method of creating ancillas for the code family using only constant overhead.  The latter would even work for a non-LDPC code, since the method for creating large ancillas would let us do Knill or Steane EC instead of Shor EC.

It is worth asking what the minimum time overhead can be, in much the same spirit that this paper has concentrated on minimizing the space overhead.  The answer is that in the same model as this paper (with free classical computation and no geometric restrictions on gates) the time overhead can also be made quite small.  It can be brought down to a constant, and a small constant at that.  Essentially standard methods of fault tolerance suffice.  Fault tolerance with concatenated codes is usually treated in a strictly self-similar way, with a level $i$ error correction injected after each level $i+1$ gate, for all $i$.  However, for the concatenated $7$-qubit code, Clifford group gates are transversal, and there is little distinction between a level $i$ Clifford group gate and a level $i+1$ Clifford group gate except on how many blocks it is performed.  In addition, the measurements used in Steane and Knill error correction give us the error syndrome for all levels simultaneously.  A non-Clifford group gate, such as the $\pi/8$ gate, can be done at any level once the appropriate ancilla is created.  Knill gates combine gates and error correction, and so seem to be the most time-efficient, with a cost of just $2$ time steps for the Bell measurement (depending on the precise physical gates available, of course).  Ancillas can be prepared in parallel to the main computation so that they are ready just when they are needed.  The non-Clifford gate may need a follow-up gate to correct for the teleportation, which needs $2$ more time steps.  Therefore, the time overhead for concatenated codes can be brought down to just a constant factor $4$.  In addition, there is an additive time overhead for the time needed to encode ancilla states at the start of the computation, but since the encoding time per level can be a constant, the time needed to fully encode an ancilla is just $O(\log \log f(k)/\epsilon)$ (using the notation of Thm.~\ref{thm:main}), which will be negligible for most computations.  The classical computation required can also be brought into constant depth per level of concatenation, i.e., $O(\log \log f(k)/\epsilon)$ at each logical time step.

It is quite reasonable to believe that there is a fundamental tradeoff between space and time, and that low qubit overhead means larger time overhead, but this does not necessarily have to be the case.  On the one hand, we can have constant space overhead and time overhead $k^{1-\gamma}$ (for $\gamma$ arbitrarily close to $1$).  On the other, we have constant time overhead and space overhead $\polylog (k/\epsilon)$.  It remains an open question whether it is possible to have simultaneously constant space and time overhead.

\section{Conclusion}
\label{sec:conclusion}

I have shown that, in principle, fault-tolerant quantum computation is possible with few extra qubits.  This is in great contrast to conventional approaches to fault tolerance.  The result depends on having no geometric constraints, on fast classical computation, and above all only works in the asymptotic limit.  A direct application of the proofs in this paper would conclude that the error threshold which allows low overhead is quite bad, but a good deal of that is based on approximations useful for most easily proving the result.  There is no reason to believe the threshold for the protocol described here is much worse than that achievable via other fault-tolerant protocols.  Indeed, the same arguments used in this paper can also show the existence of a threshold for the toric code, yet simulations suggest that the threshold for fault-tolerant protocols for the toric code is much better than would be implied by the proofs used here~\cite{RH07}.  This provides a reason to hope that the threshold for other LDPC codes could be similarly high.  However, there is a direct tradeoff between tolerance of storage errors and the overhead parameter $\eta$ in Thm.~\ref{thm:main}.  The overhead for modest-size implementations of the protocol will be much worse than the asymptotic rates, but possibly still far better than for most existing fault-tolerant protocols.

The main practical drawback of the protocol in this paper is that we do not know of a family of quantum LDPC codes that has all the properties we want (LDPC, exponential error suppression, efficient decoding).  This is an important open question that deserves attention.   It is conceivable that subsystem codes might be useful: The only place the LDPC condition is used in the proof of Thm.~\ref{thm:main} is to perform error correction with few extra ancilla qubits.  Suppose we have a subsystem code for which the stabilizer generators have high weight, but each generator can be written as a product of constant-weight gauge operators.  Then measuring the gauge operators and multiplying the results together appropriately will tell us the error syndrome.  However, it seems unlikely that an error correction procedure of this form can be made robust against errors during error correction; for a true LDPC code, the value of each measured syndrome bit is stable unless there is an error on the qubits involved in that syndrome bit, but gauge generators do not commute with each other, so we do not expect them to retain the same eigenvalue under repeated measurements even in the absence of error.  Therefore, it will be very hard to distinguish whether a change in a gauge operator eigenvalue is due to an error or not, making it difficult to identify syndrome errors.

Unfortunately, the defects of the current code families also make it difficult to study numerically the actual performance of this protocol.  On the one hand, for the cases with an exponential decoding algorithm, determining logical error rates for even modest-sized codes requires an enormous amount of computation.  On the other hand, if we use a hyperbolic geometry code with an efficient decoding algorithm, errors are only weakly suppressed, meaning we have to go to very large codes before the logical error rates are low enough to be useful.  Simulating sufficiently large codes is again a challenging computational task.

It may also be possible to improve on this protocol by finding better families of codes that either have a way to perform fault-tolerant gates with small ancillas or have a better way of creating the large ancillas needed.  Alternatively, other approaches to efficient fault tolerance might exist, maybe even some that work well with geometric locality constraints on the physical gates.  The main thing is not to give up: there is no inherent reason that quantum fault tolerance needs large numbers of extra qubits, and we should not be satisfied with protocols that do.

\paragraph{Acknowledgements}
I would like to thank Alexey Kovalev for clarifying some aspects of \cite{KP12b} and for discussions on how to apply the result in the context of a full fault-tolerant protocol, and I would like to thank Nicolas Delfosse and Matt Hastings for discussion of hyperbolic codes.  Research at Perimeter Institute is supported by the Government of Canada through Industry Canada and by the Province of Ontario through the Ministry of Research and Innovation.  Part of this paper was written while the author was a guest at the Newton Institute for Mathematical Sciences.

\end{document}